\documentclass[10pt, conference, letterpaper]{IEEEtran}
\usepackage{url}
\usepackage{cite}
\usepackage[noend]{algpseudocode}
\usepackage{graphicx}
\usepackage{xcolor}
\usepackage{caption}
\usepackage{amsthm}
\usepackage{amsmath}
\usepackage{amssymb}
\usepackage{subfig}
\usepackage{blindtext}
\usepackage{balance}
\usepackage{comment}
\usepackage{siunitx}
\usepackage[linesnumbered,ruled,noend]{algorithm2e}

\newcommand{\set}[1]{\left\{#1\right\}}

\newtheorem{theorem}{Theorem}[section]
\newtheorem{lemma}[theorem]{Lemma}

\newtheorem*{observation*}{Observation}

\newcommand{\red}[1]{\textcolor{red}{#1}}
\newcommand{\blue}[1]{\textcolor{blue}{#1}}

\algtext*{EndWhile}
\algtext*{EndIf}

\ifdefined\submissionVersion

\fi

\usepackage{tikz}
\usetikzlibrary{arrows,patterns}
\usepackage{mathtools}
\usepackage{pgfplots}
\pgfplotsset{compat=newest}
\pgfplotsset{overwrite option/.style args={#1 with #2}{#1=#2,#1/.code=}}

\begin{document}

\date{}

\title{ Efficient Measurement on Programmable Switches Using Probabilistic Recirculation } 

\author{
\IEEEauthorblockN{Ran Ben Basat} 
\IEEEauthorblockA{Technion  \\ {\small sran@cs.technion.ac.il}
}
\and
\IEEEauthorblockN{Xiaoqi Chen}
\IEEEauthorblockA{Princeton University \\ {\small xiaoqic@cs.princeton.edu}}
\and
\IEEEauthorblockN{Gil Einziger} 
\IEEEauthorblockA{Nokia Bell Labs \\ {\small gil.einziger@nokia.com}}
\and
\IEEEauthorblockN{Ori Rottenstreich}
\IEEEauthorblockA{Technion \\ {\small ori.rot@gmail.com}}
}

\maketitle

\ifdefined\withoutComments
\newcommand{\ran}[1]{}
\newcommand{\danny}[1]{}
\newcommand{\gil}[1]{}
\newcommand{\ori}[1]{} 
\else
\newcommand{\ran}[1]{\blue{Ran: #1}}
\newcommand{\danny}[1]{\red{Danny: #1}}
\newcommand{\gil}[1]{\blue{ Gil: #1}}
\newcommand{\ori}[1]{\red{ Ori: #1}} 
\fi



\begin{abstract}
Programmable network switches promise flexibility and high throughput, enabling applications such as load balancing and traffic engineering.
Network measurement is a fundamental building block for such applications, including tasks such as the identification of heavy hitters (largest flows) or the detection of traffic changes.

However, high-throughput packet processing architectures place certain limitations on the programming model, such as
restricted branching, limited capability for memory access, and a limited number of processing stages.
These limitations restrict the types of measurement algorithms that can run on programmable switches.
In this paper, we focus on the RMT programmable high-throughput switch architecture, and carefully examine its constraints on designing measurement algorithms. We demonstrate our findings while solving the heavy hitter problem.

We introduce PRECISION, an algorithm that uses \emph{Probabilistic Recirculation} to find top flows on a programmable switch. 
By recirculating a small fraction of packets, PRECISION simplifies the access to stateful memory to conform with RMT limitations and achieves higher accuracy than previous heavy hitter detection algorithms that avoid recirculation. 
We also analyze the effect of each architectural constraint on the measurement accuracy and provide insights for measurement algorithm designers.
\end{abstract}

\newcommand{\matrixCellWidth}{5.8cm}

\section{Introduction}
\label{sec:intro}
Programmable network switches enable rapid deployment of network algorithms such as traffic engineering, load balancing, quality-of-service optimization, anomaly detection, and intrusion detection~\cite{ApproximateFairness,IntrusionDetection,TrafficEngeneering,IntrusionDetection2,LoadBalancing}.

Measurement capabilities are often at the core of such applications, as they extract information from traffic to make informed decisions~\cite{BetterNetflow}. 
Typically there are millions of network flows to monitor~\cite{CounterArray1,CounterArray2} and ideally each flow is allocated some memory for storing its measurement statistics. Coping with the 100~Gbps line rate requires the measurement to be stored in fast but limited-capacity SRAM memory. However, SRAM is too small for keeping exact flow statistics for every flow. 
\emph{Heavy Hitter} algorithms only store flow state for the largest flows to overcome this limitation. 
This approach exposes a trade-off between memory space and accuracy, where additional space improves the accuracy. 

There are two types of solutions for heavy hitter detection problem  ---
\emph{counter-based} algorithms and \emph{sketch-based} algorithms.
Counter-based algorithms maintain a bounded-size flow cache. Only a small portion of the flows are measured, and each monitored flow has its own counter~\cite{LC,SpaceSavingIsTheBest2010}.
Examples of counter-based algorithms include \emph{Lossy Counting}~\cite{LC}, \emph{Frequent}~\cite{BatchDecrement}, \emph{Space-Saving}~\cite{SpaceSavings}, and \emph{RAP}~\cite{RAP}. In sketch-based algorithms, counters are implicitly shared by many flows. Examples of sketch-based algorithms include 
\emph{Multi Stage Filters}~\cite{CUSketch}, \emph{Count-Min Sketch}~\cite{CountMinSketch}, \emph{Count Sketch}~\cite{CountSketch},  
\emph{Randomized Counter Sharing}~\cite{RandomizedCounterSharing}, 
\emph{Counter Tree}~\cite{CounterTree}, and \emph{UnivMon}~\cite{UnivMon}.


Heavy hitter measurement has two closely related goals.
In the \emph{frequency estimation} problem, we wish to approximate the size of a flow whose ID is given at query time.
Alternatively, in the \emph{top-k} problem, the algorithm is required to list the $k$ top flows.  In general, sketch algorithms solve the frequency estimation problem but require additional efforts to address the top-$k$ problem. 
For example, UnivMon~\cite{UnivMon} uses heaps alongside the sketches to track the top flows.
FlowRadar~\cite{FlowRadar} and Reversible~Sketch~\cite{ReversibleSketch} encode flow ID in the sketch, and have a small probability to fail to decode. 
In contrast, counter algorithms already store flow identifiers and can directly solve the top-$k$ problem.  While sketch algorithms are readily implementable in programmable switches, supporting top-$k$ measurements is a strong motivation for deploying counter algorithms in such switches. Unfortunately, high-performance packet processing imposes severe restrictions on the programming model which makes implementing counter algorithms~a~daunting~task. 




\textbf{Contribution. }
Our work summarizes the restrictions of the Reconfigurable Match Tables (\emph{RMT})~\cite{RMT} switch programming model in the context of measurement algorithm design.
The RMT breakthrough design allows a pipeline multiple match-action tables of different widths and depths and was recently described as a ``key enabling technology for the support of the emerging P4 data plane programming''~\cite{dargahi2017survey}.
We divide the RMT restrictions into four easy-to-understand rules: \emph{limited branching}, \emph{limited concurrent memory access}, \emph{single stage memory access}, and \emph{fixed number of stages}. 
%
%
%
%

We present \emph{Probabilistic RECirculation admisSION ({PRECISION})} -- a heavy hitter algorithm that is fully compatible with the RMT high-performance programmable switch architecture. PRECISION is implemented in the P4 language and can be compiled to the newly released Barefoot Tofino~\cite{Tofino} programmable switch that achieves multiple Tbps of aggregated throughput.
The P4 source code of PRECISION can be found at~\cite{GitHub}.
The core idea behind PRECISION is \emph{probabilistic recirculation}; PRECISION randomly recirculates a small portion of packets from unmonitored flows; when a packet is recirculated, it passes through the programmable switching pipeline twice.
In the first pipeline pass, we try to match a packet to an existing flow entry; if this succeeds, we increment its counter. If unmatched, we probabilistically recirculate it
to claim an entry with the new packet's flow~ID.
Using the packet recirculation feature greatly simplifies the memory access pattern without significantly degrading throughput, while by carefully setting the recirculation probability we can achieve high monitoring accuracy.

Previous suggestions include HashParallel and HashPipe  \cite{HashPipe}, two counter-based heavy hitter detection algorithms proposed specifically for running on high-throughput programmable switches. They both maintain a $d$-stage flow table tailored to the pipeline architecture of programmable switches but differ in whether to recirculate an unmatched packet.
HashPipe never recirculates packets and always inserts the new entry, which yields high throughput but lower accuracy. 
Instead, HashParallel recirculates \textbf{every} unmatched packet, which achieves much better accuracy but lowers the throughput. In contrast, PRECISION only recirculates a tiny portion of the unmatched packets with a minimal impact on performance. This approach allows PRECISION to conform to the RMT memory access restrictions and also improves accuracy over HashPipe, especially for heavy-tailed workloads. We then analyze the impact of each RMT constraint individually and find that most limitations have little effect in practice. 
We also show that HashPipe~\cite{HashPipe} cannot satisfy both the 
\emph{limited branching} rule and the \emph{single stage memory access} rule, and is, therefore, challenging to implement in the RMT model.
This highlights the importance of our study of the model limitations for researchers and practitioners alike. 


Finally, we evaluate PRECISION on real packet traces and show that it improves on the state-of-the-art for high-performance programmable switches (HashPipe) for the two variants of the heavy hitter problem. It is up to 1000~times more accurate than HashPipe for the frequency estimation problem and reduces the space required to correctly identify the top-128 flows by a factor of up to 32~times. When compared to general (software) heavy hitter algorithms, PRECISION often has similar accuracy compared to Space-Saving and RAP. Interestingly, approximating the desired recirculation probability appears very important, with a stage-efficient 2-approximate solution PRECISION requires at most four times as much memory as RAP. When we dedicate more hardware pipeline stages to achieve a better approximation, the performance gap between PRECISION and RAP diminished. 

\textbf{Paper Outline. }
The paper is structured as follows.
Section~II outlines the programming restrictions of the RMT high-performance switch architecture and their impact on designing data plane algorithms.
Section~III introduces the reader to the heavy hitter detection problem and surveys related work.
Section~IV discusses the implementation of PRECISION, specifically how we adapt to the limitations imposed by the RMT architecture.
We present theoretical analysis on bounding the amount of recirculation in Section~V.
In Section~VI, we evaluate PRECISION, by first quantifying the impact of each adaptation on the accuracy, and then position it within the field by comparing it with other heavy-hitter detection algorithms. Finally, we conclude in Section~VII. 

\section{Constraints Of Programmable Switches}
\label{sec:hardware}
The emergence of P4-based programmable data plane~\cite{P4} is an exciting opportunity to push network algorithms to programmable switches.
In this section, we give a brief introduction of the recently developed RMT~\cite{RMT} high-performance programmable switch architecture and then explain its programming model and its restrictions in the context of network measurement algorithm design. 

The RMT architecture uses a pipeline to process packets. At a glance, the packet first goes into a programmable packet header parser that extracts the header fields, and then traverses a series of pipeline stages, and finally is emitted through a deparser.
Each stage includes a \emph{Match-Action Table}, which first performs a \emph{Match} that reads some packet header fields and matches them with a list of values. Then, it performs the corresponding \emph{Action} on the packet, which can be routing decisions or modifying header field variables.
RMT promises Tbps-level throughput which is achieved by limiting the complexity of pipeline stages. These typically run at a fixed clock cycle $\geq \SI{1}{\GHz}$ (i.e., $\leq 1\mu s$ processing time), permitting only elementary actions.
Flexibility is achieved by allowing many parallel actions in the same stage, and by connecting many simple stages into a pipeline. 

Our case-study of heavy hitter measurement in this model exposed the following restrictions which we survey below. 
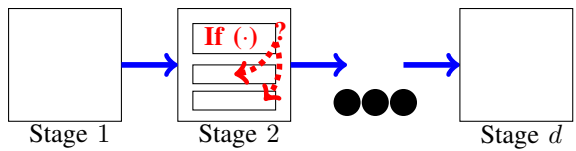
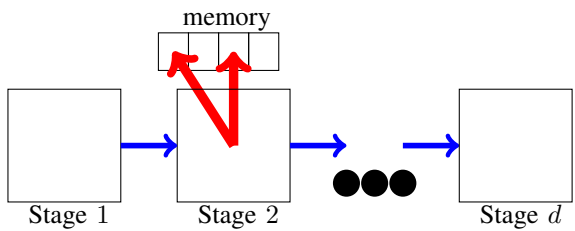
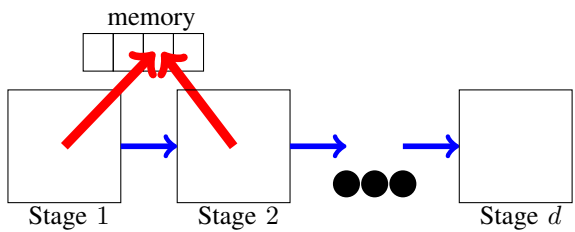
\begin{figure}
	\centering
\subfloat[Restriction I: Limited in-stage branching]{
\label{fig1:res-branching}
\begin{tikzpicture}[scale=0.5]
\tikzstyle{every loop}=[]
\draw (0,0) rectangle (3,3);
\filldraw[black] (0.3,-0.4) circle (0pt) node[anchor=west] {Stage $1$};

\draw (4.5,0) rectangle (7.5,3);

\draw (4.9,1.8) rectangle (7.1,2.6);
\filldraw[black] (4.95,2.2) circle (0pt) node[anchor=west] {\textcolor{red}{\textbf{If ($\cdot$)}}};
\draw (4.9,1.0) rectangle (7.1,1.5);
\draw (4.9,0.3) rectangle (7.1,0.8);

  \draw [->,dotted,thick, red, line width=0.75mm] (7.1,2.1) to [bend left] (6,1.25);
   \draw [->,dotted,thick, red, line width=0.75mm] (7.1,2.1) to [bend left] (7-0.2,0.55);
   \filldraw[black] (6.8,2.4) circle (0pt) node[anchor=west] {\textcolor{red}{\textbf{?}}};
   
\filldraw[black] (4.8,-0.4) circle (0pt) node[anchor=west] {Stage $2$};
 \draw [->,solid,thick, blue, line width=0.75mm] (3,1.5) to  (4.5,1.5);
 
\node[style={circle,draw,thick,align=center, fill={black}}] (0) at (9.0, 0.5) {};
\node[style={circle,draw,thick,align=center, fill={black}}] (0) at (9.75, 0.5) {};
\node[style={circle,draw,thick,align=center, fill={black}}] (0) at (10.5, 0.5) {};
 \draw [->,solid,thick, blue, line width=0.75mm] (7.5,1.5) to  (9,1.5);
 \draw [->,solid,thick, blue, line width=0.75mm] (10.5,1.5) to  (12,1.5);
\draw (12,0) rectangle (15,3);

\filldraw[black] (12.3,-0.4) circle (0pt) node[anchor=west] {Stage $d$};
 \path[]
 ;
\end{tikzpicture}
}

	\subfloat[Restriction II: Limited memory access within a stage]{
\label{fig1:res-way}
\begin{tikzpicture}[scale=0.5]
\tikzstyle{every loop}=[]
\draw (0,0) rectangle (3,3);
\filldraw[black] (0.3,-0.4) circle (0pt) node[anchor=west] {Stage $1$};

\draw (4,3.5) rectangle (4.8,4.5);
\draw (4.8,3.5) rectangle (5.6,4.5);
\draw (5.6,3.5) rectangle (6.4,4.5);
\draw (6.4,3.5) rectangle (7.2,4.5);
\filldraw[black] (4.4,4.8) circle (0pt) node[anchor=west] {memory};
  \draw [->,solid,thick, red, line width=1.25mm] (6,1.5) to  (4.4,4);
   \draw [->,solid,thick, red, line width=1.25mm] (6,1.5) to  (6,4);
\draw (4.5,0) rectangle (7.5,3);

\filldraw[black] (4.8,-0.4) circle (0pt) node[anchor=west] {Stage $2$};
 \draw [->,solid,thick, blue, line width=0.75mm] (3,1.5) to  (4.5,1.5);
   
\node[style={circle,draw,thick,align=center, fill={black}}] (0) at (9.0, 0.5) {};
\node[style={circle,draw,thick,align=center, fill={black}}] (0) at (9.75, 0.5) {};
\node[style={circle,draw,thick,align=center, fill={black}}] (0) at (10.5, 0.5) {};
 \draw [->,solid,thick, blue, line width=0.75mm] (7.5,1.5) to  (9,1.5);
 \draw [->,solid,thick, blue, line width=0.75mm] (10.5,1.5) to  (12,1.5);
\draw (12,0) rectangle (15,3);

\filldraw[black] (12.3,-0.4) circle (0pt) node[anchor=west] {Stage $d$};
 \path[]
 ;
\end{tikzpicture}
}

	\subfloat[Restriction III: One memory address cannot be accessed from multiple stages]{
\label{fig1:res-hazard}
\begin{tikzpicture}[scale=0.5]
\tikzstyle{every loop}=[]
\draw (0,0) rectangle (3,3);
\filldraw[black] (0.3,-0.4) circle (0pt) node[anchor=west] {Stage $1$};

\draw (2,3.5) rectangle (2.8,4.5);
\draw (2.8,3.5) rectangle (3.6,4.5);
\draw (3.6,3.5) rectangle (4.4,4.5);
\draw (4.4,3.5) rectangle (5.2,4.5);
\filldraw[black] (2.4,4.8) circle (0pt) node[anchor=west] {memory};
  \draw [->,solid,thick, red, line width=1.25mm] (1.5,1.5) to  (3.9,4);
   \draw [->,solid,thick, red, line width=1.25mm] (6,1.5) to  (4.1,4);
\draw (4.5,0) rectangle (7.5,3);

\filldraw[black] (4.8,-0.4) circle (0pt) node[anchor=west] {Stage $2$};
 \draw [->,solid,thick, blue, line width=0.75mm] (3,1.5) to  (4.5,1.5);
   
\node[style={circle,draw,thick,align=center, fill={black}}] (0) at (9.0, 0.5) {};
\node[style={circle,draw,thick,align=center, fill={black}}] (0) at (9.75, 0.5) {};
\node[style={circle,draw,thick,align=center, fill={black}}] (0) at (10.5, 0.5) {};
 \draw [->,solid,thick, blue, line width=0.75mm] (7.5,1.5) to  (9,1.5);
 \draw [->,solid,thick, blue, line width=0.75mm] (10.5,1.5) to  (12,1.5);
\draw (12,0) rectangle (15,3);

\filldraw[black] (12.3,-0.4) circle (0pt) node[anchor=west] {Stage $d$};
 \path[]
 ;
\end{tikzpicture}
}
 \caption{Illustration of some restrictions imposed by RMT pipeline model for designing measurement algorithm.
	}\label{fig_blockchain}
	\vspace*{-5mm}
\end{figure}

\textbf{Simple Per-stage Actions (Limited Branching).} 
Each pipeline stage can only execute primitive arithmetic. For example, division is much slower than addition; thus the switching hardware usually does not support division. Also, branching operations are expensive, and the hardware pipeline may only support very limited branching within stages (but can perform complex branching across stages), as illustrated in Figure~\ref{fig1:res-branching}.
Therefore, we cannot perform arbitrary computation and have to redesign the algorithm to fit the architecture.

\textbf{Limited Concurrent Memory Access. } 
A small amount of static random access memory (SRAM) is attached to each hardware stage for stateful processing. As illustrated in Figure~\ref{fig1:res-way}, when a packet arrives, it can access one, or a few, addresses in the memory region but not read or write the entire memory block, again due to per-stage timing requirement.
From an algorithm design perspective, this means we can only read from or write to memory at specific addresses, and therefore cannot compute even the most straightforward functions globally, e.g., find a minimum across many~array~elements.

\textbf{Single Stage Memory Access. } 
Each stage is processing a different packet at any given time. Therefore, allowing two packets to access the same memory region may cause a read-write hazard, shown in Figure~\ref{fig1:res-hazard}. The RMT architecture avoids this by allowing access to stateful memory blocks only from one particular pipeline stage.  Thus, our algorithm can only access each memory region once as the packet is going through the pipeline. We need to \emph{recirculate} a packet, causing it to go through the entire pipeline again, in order to access the same memory block twice. Recirculation is expensive as it reduces the rate that incoming packets can access the pipeline.  

Even in more recently proposed architecture like dRMT~\cite{dRMT} where memory resources are dynamically allocated to different hardware stages, we still cannot allow accessing the same memory region from two different pipeline stages.  Therefore, the restriction we describe seems fundamental. 

\textbf{Fixed Number of Stages. } For guaranteeing a low per-packet latency, the switch cannot have too many pipeline stages. In our case, since the pipeline is not very long, the total number of operations performed on a packet cannot exceed a hardware-imposed constant. Again, we can circumvent the limit by recirculating some packets, with a throughput impact.


\textbf{Discussion.}
While these restrictions target specifically the newly proposed RMT architecture, we believe that future high-throughput switching architectures are likely to have similar constraints due to the throughput and latency requirements they need to satisfy.  

We also note that capabilities prepared for packet forwarding can be exploited by measurement algorithms as well.
The Match-Action Table model specifies that each pipeline stage will use a part of packet header data (e.g., a network address) to perform a lookup in a match table, and subsequently executes the corresponding action in the table (e.g., a forwarding destination).
In our algorithm design perspective, this means we can perform parallel lookups on intermediate variable cheaply. Beyond \emph{exact} matching, the architecture also supports \emph{ternary} and~\emph{longest-prefix}~matching.

Note that the TCAM memory used in table lookup is different from the memory used for stateful processing (SRAM) mentioned earlier. 
TCAM allows for parallel reads, but writing may not finish in constant time. Hence it can only be modified by the switch control plane but not within the data plane (by the packet being processed, in one pipeline clock cycle). Thus, the parallel-readable lookup tables are ``read-only'' for the packet, and writable memory must be accessed by addresses.

\section{Problem Definition and Existing Solutions} 
\label{sec:related}
This section formally defines the problems addressed in this work as well as surveys the relevant related work. 

\subsection{Problem statement}
Our work targets two common measurement forms, the \emph{frequency estimation} problem and the \emph{top-$k$} problem. For both, we refer to a quasi-infinite packet stream, where each packet is associated with a flow as explained~below. 

A flow refers to a particular subset of the packet stream that we choose to combine and analyze as a whole. For example, a flow may apply to a TCP or a UDP connection, in which case the connection five-tuple (source and destination IP, protocol, source and destination port) becomes the flow identifier. Alternatively, a flow may refer to just the source IP address, or just the destination IP and port pair. In any case, we assume that a flow identifier is available from some fields of the packet header, and that flows partition the stream such that each packet belongs to a single flow.

We denote the frequency of a network flow with ID $s$, or the total number of packets belonging to flow $s$, as $f_s$. For the \emph{frequency estimation} problem, we use the OnArrival 
model~\cite{RAP}, 
which requires an algorithm to estimate the flow frequency for each new packet it sees, and evaluates the estimation error upon each packet arrival. 
Formally, we reveal packets in a stream $(p_1,p_2,\ldots)$ one packet at a time, and on each packet arrival, with packet $p_t$ belonging to some flow $s$. An algorithm \emph{Alg} is required to provide an estimate $\widehat{f_{s}}$ for 
$f_{s}\triangleq|\left\{  
p_i \in s|1\leq i \leq t
\right\}|$ ---
the number of packets belonging to flow $s$ in $p_1, \dots, p_t$.
We then measure the \emph{Mean Square Error (MSE)} of the algorithm, i.e.,
$$MSE(Alg) \triangleq \frac{1}{N}\sum_{t=1}^{N}(\widehat{f_{s}} - f_{s})^2.$$

The \emph{top-$k$} identification problem is defined as follows: Given a stream $(p_1,p_2,\ldots)$ and a query parameter $k$, the algorithm outputs a set of flows containing as many of the $k$ largest flows as possible.
We denote the $k^{th}$ largest flow's frequency by $F_k$.
When the algorithm outputs a flow set $\mathbf{C}$,
we judge its quality using the standard \emph{Recall} metric that measures how many top flows it identifies:
\begin{align*}
\text{Recall}(\mathbf{C})    &\triangleq\ \  |e\in \mathbf{C} : f_e \ge F_k| / k.
\end{align*}

\subsection{Existing Approaches}
\textbf{The Space-Saving algorithm:}
Space-Saving (SS)~\cite{SpaceSavings} is a heavy hitter algorithm designed for database applications and software implementations.
Space-Saving maintains a fixed-size flow table, where each entry has a flow identifier and a counter. 
When a packet from an unmonitored flow arrives, the identifier of the minimal table entry is replaced with the new flow's identifier, and its counter is incremented. Space-Saving uses a sophisticated data structure named stream-summary which allows it to maintain the entries ordered according to counter values in constant time as long as all updates are~of~unit~weight. 

Space-Saving was designed for database workloads, which often exhibit a heavily concentrated access pattern, i.e. most of the traffic comes from a few heavy hitters.
In contrast, networking traces are often heavy-tailed~\cite{FilteredSpaceSaving,RAP}. That is, a non-negligible percentage of the packets belong to tail flows or those other than heavy hitters. Unfortunately, Space-Saving works poorly on such workloads. 

\textbf{Optimization for heavy-tailed workloads:}
To deal with heavy-tailed workload,
Filtered Space-Saving~\cite{FilteredSpaceSaving} 
attempts to filter out tail flows before inserting into flow table.
It utilizes a bitmap alongside a Space-Saving instance.
When a packet arrives, a hash function is used to map its flow ID into a bitmap entry. If the entry is zero, it merely sets the entry to one. Otherwise, we update the Space-Saving~instance.


Maintaining additional data structures to filter tail flows may be wasteful. Therefore, \emph{Randomized Admission Policy (RAP)}~\cite{RAP} suggests using randomization instead.
When an unmonitored flow arrives, it is admitted only with a small probability. Thus, most tail flows are filtered while heavy hitters that appear many times are eventually admitted. 
%
Specifically, if the minimal entry has a counter value of $c$, RAP requires the competing flow to win a coin toss with a probability of $\frac{1}{c+1}$ to be added. 
The idea of RAP can be applied to the Space-Saving algorithm for software implementations. For hardware efficiency, the authors evaluate a limited associativity variant. 

Unfortunately, the programming model of high-performance programmable switches is too restrictive to implement these algorithms directly. 
Specifically, Space-Saving evicts the minimal flow entry across all monitored flows, whereas the architecture of programmable switches does not permit finding (and replacing) the minimum element among all counters. 
Even for the limited associativity variant of RAP, it is still difficult to implement the randomize replacement after finding the approximate minimum value, due to same-stage memory access restriction.

\textbf{High-performance switch algorithms:} 
HashPipe~\cite{HashPipe} adapts Space-Saving to meet the design constraints of the P4 language and PISA programmable switch architecture~\cite{P4}. 
The authors suggest partitioning the counters into $d$ separate stages to fit the programmable switch pipeline. They use $d$ hash functions that dictate which counter can accommodate each flow on each stage.
They first propose a strawman solution, \emph{HashParallel}, which makes each packet traverse all $d$ stages while tracking the minimal value among the counters associated with its flow. If the flow is monitored, HashParallel increments its counter. If not,  it recirculates the packet to replace the minimal entry among the $d$. The authors explain that HashParallel potentially recirculates all the packets, which halves the throughput. 

Hence, they suggest HashPipe as a practical variant with no recirculation. In HashPipe, each packet's flow entry is always inserted in the first stage. 
They then find a rolling minimum --- the evicted flow proceeds to the next stage where its counter is compared with the flow monitored there. The flow with the larger counter remains, while the smaller flow's entry is propagated further. Eventually, the smaller counter on the $d^{th}$~stage~is~evicted. This allows HashPipe to avoid recirculation but introduces the problem of duplicates --- some flows may occupy multiple counters, and small flows may  still evict other flows.

Flow Radar~\cite{FlowRadar} is another P4 measurement algorithm that follows a different design pattern. The main design difficulty to overcome is the lack of access to a fully associative hash table in programmable switches. While HashPipe and this work implement a fixed associativity table using multiple pipeline stages, FlowRadar potentially stores multiple flows within the same table entry. That is, upon hash collision the new flow identifier is XORed into the existing identifier. FlowRadar works best when the measurement is distributed, where multiple programmable switches can share their state to decode flow entries. Initially, FlowRadar recovers all flow entries that had no collision. Recovered flows are then recursively removed from the data structure, enabling for more flows to be recovered. 

This approach is differentiated from our own as it attempts to perform an accurate measurement and therefore requires space which is proportional to the number of flows. In contrast, our approach provides an approximation of the flow sizes, and the required memory is independent of the number of flows.  Also, FlowRadar requires multiple measurement devices each encoding a different subset of flows whereas our solution can also be implemented on a single device.

\textbf{Sampling:} 
Instead of running algorithms in data plane, one may also sample a fraction of packets and run sophisticated algorithms elsewhere. This approach simplifies the hardware implementation but the problem migrates elsewhere. Namely, to process the samples in real time, we need additional computation and bandwidth overheads. Also, achieving high monitoring accuracy on smaller flows requires~high~sampling~rate.



\section{Design and Implementation of PRECISION}
\label{sec:impl}
We now present several hardware-friendly adaptations that address the restrictions imposed by RMT switch architecture. 

\subsection{From fully associative to $d$-way associative memory access}
Building on top of Space-Saving~\cite{SpaceSavings} and RAP~\cite{RAP}, we first tackle the fact that a programmable switch cannot perform the fully-associative memory access to evict the minimum item.
At any given pipeline stage, the algorithm can specify an index to access some location in the register array. The switch may allow accessing a small number of positions simultaneously but definitely cannot compute a global minimum across an entire register array.

We adopt the limited-associativity idea from HashParallel and HashPipe~\cite{HashPipe} to approximately evict a small element, by choosing the minimum across $d$ randomly selected elements from $d$ separate register arrays. With this relaxation, we can naturally spread the memory access across different hardware stages, and at each hardware stage, we only access one memory location. Specifically, we use $d$ independent hash functions $h_1,\dots,h_d$ to compute a different index for each stage, and at each stage, we access the $h_i(key)^{th}$ element of the $i^{th}$ register array. Note that PRECISION performs $d$ flow entry reads, but it does not consume exactly $d$ hardware pipeline stages, as processing each read involves two branchings, and costs three hardware stages. We also discuss how to reduce the total number of hardware stages required in Section~\ref{sec:parallel}.

\subsection{Simplified memory access} 
\paragraph{Why HashPipe violates RMT?}
Although the design of HashPipe has already satisfied many restrictions imposed by RMT structure, its memory access pattern prevents us from implementing it in today's programmable switch hardware.
The high-level idea of the HashPipe algorithm (see pseudocode in Algorithm~\ref{algo:HashPipe}) is to always evict the minimum out of $d$ elements, by ``carrying" a candidate eviction element through the pipeline. 
At each stage, we compare the counter read from register memory with that of the carried element. Then, the smaller of which is propagated further onward.

\begin{algorithm}[t]\small
$l_1 \gets h_1(iKey)$\Comment{Always insert in the first stage}\;
\If{$key_1[l_1] = iKey$}{
  $val_1[l_1] \gets val_1[l_1] + 1$\;
  end processing\;
}\ElseIf{$l_1$ is an empty slot}{
  $(key_1[l_1], val_1[l_1]) \gets (iKey, 1)$\;
  end processing\;
}\Else{
  $(cKey, cVal) \gets (key_1[l_1], val_1[l_1])$\;
  $(key_1[l_1], val_1[l_1]) \gets (iKey, 1)$\;
}
\Comment{Track a rolling minimum}\;
\For{$i \gets 2$ \textbf{to} $d$} {
$l_i \gets h_i(cKey)$\;
  \If{$key_i[l_i] = cKey$ } { \label{algline:reg1}  \Comment{Read \red{$key_i$}} \;
    $val_i[l_i] \gets val_i[l_i] + cVal$\;  \Comment{R/W \blue{$val_i$} }\;
    end processing\;
  }
  \ElseIf{$[l_i]$ is an empty slot}{
    $(key_i[l_i], val_i[l_i]) \gets (cKey, CVal)$ \; \Comment{Read  \red{$key_i$}, \blue{$val$} }\;
    end processing\;
  }
  \ElseIf{$val_i[l_i] < cVal$ \label{algline:reg2} }{ \Comment{Condition on \blue{$val_i$}; \red{Violating Restriction I}}\;
    swap $(cKey, cVal) \Leftrightarrow (key_i[l_i], val_i[l_i])$ \label{algline:reg3} \; \Comment{R/W \red{$key_i$} }\;
  }
}
\caption{HashPipe~\cite{HashPipe} heavy hitter algorithm}
\label{algo:HashPipe}
\end{algorithm}

We now scrutinize the register memory access to different arrays of HashPipe, as highlighted  in Algorithm~\ref{algo:HashPipe}.
If we look at Line~\ref{algline:reg1} and Line~\ref{algline:reg3}, they both access the register array $key$ holding flow identifiers. The single stage memory access restriction requires that line~\ref{algline:reg1} through line~\ref{algline:reg3} would be placed within the same hardware pipeline stage. 

However, the execution flow is branched in line~\ref{algline:reg2} based on the values in another register array $(val)$. Such branching violates the limited branching restriction and may not be easily implemented within a single hardware pipeline stage in today's programmable switches. 
Referring to the model presented in \cite{Domino}, to implement HashPipe, 
the simple \emph{RAW} %
\footnote{The RAW action unit is capable of Reading an element from register memory, Add a value to it, and Write it back. See \cite{Domino}.}
action atoms at each stage are inadequate, and at least 
\emph{Paired} %
\footnote{The Paired action unit is capable of reading two different elements from register memory, conditionally branch twice (two nested \emph{if}s), perform addition or subtraction to the elements, and write two new values back. See \cite{Domino}.} 
action atoms are required.
While the RMT architecture~\cite{RMT} does not specifically define what features the action units need to support, Paired action atoms are more expensive to implement than RAW atoms and require 14x larger chip area than RAW atoms~\cite{Domino}.  We strive to design our measurement algorithm to only require the simpler RAW atoms.

Without such atoms, it is difficult to conditionally update a flow entry while simultaneously incrementing the corresponding counters. As long as we place flow identifier and counter in two separate register arrays, this seemingly innocuous set of operations has some inevitable in-stage branching: if we access flow identifiers first, we need to:
(i) Read flow ID from flow entry array;
(ii) If ID matched, increment counter; otherwise, compute some condition on the counter;
(iii) If the condition is satisfied, replace flow ID.
This leads to a write to flow entry register memory conditioned on reading from another counter register memory. Therefore, branching within the stage is inevitable.

Some may argue that we can cleverly rearrange the operations to mitigate the branching; however, even if we access the counter first, we still encounter the same restriction:
(i) Read a counter from the counter register memory;
(ii) Read flow ID; if ID not matched, use the counter to decide whether to replace flow ID;
(iii) Write the incremented counter value, if the ID matched.
Again, the conditional write after reading another register forces a branching within a hardware pipeline stage, making it challenging to implement in today's programmable switches.

\paragraph{PRECISION's solution}
The implementation of PRECISION is even more challenging. 
We decide to replace an entry after knowing the minimum sampled counter value, but we only know this value after reaching the end of the pipeline, at which point it is too late to write to the register memory of earlier stages.

We resolve this difficulty using the recirculation feature on switches \cite{P4,Arista}, that allows packets to traverse the pipeline again, removing all conditional branching for register access. 
When a packet leaves the last stage of the pipeline, instead of leaving the switch, we may choose to bring it to the beginning of the pipeline and go through all stages again.  We can use metadata to distinguish between recirculated packets (which should be dropped) and regular packets that should be forwarded to their next hop. 

Using recirculation allows more versatile packet processing at the cost of packet forwarding performance, as the recirculated packet will compete for resources with new incoming packets. However, we believe it's a necessary trade-off to satisfy the no-branching-within-stage constraint for high-performance programmable switches.

At the end of the pipeline, we ignore those packets already matched to flow entries and probabilistically recirculate the other packets using probability $\frac{1}{carry\_min+1}$, where $carry\_min$ is the value of the minimum sampled entry. The recirculated packet will evict and replace the minimum sampled entry. It will traverse the pipeline again to write its flow identifier into the corresponding register array when it arrives at the right pipeline stage, and also update the corresponding counter to a new value $carry\_min+1$.  In expectation, for every unmatched packet we increased the count for its flow by 1.

As a packet recirculates, it introduces a delay between the point in which we chose to admit it, and when it writes its flow ID on its second pipeline traversal.
During this period other packets may increment the counter, an effect that will be overridden. Thus, the recirculation delay may have some impact on PRECISION's accuracy. 
The duration of such delay is architecture-specific and depends on both the queuing before entering the pipeline and the length of the pipeline. In Section~\ref{compromize:packetDelay}, we evaluate its impact on PRECISION's accuracy and show that PRECISION is insensitive to such delay. 

\subsection{Efficient recirculation}
We avoid packet reordering and minimize application-level performance impact by using the \textit{clone-and-recirculate} primitive, which routes the original packet out of the switch as usual, and drops the cloned packet after it finishes the second pipeline~traversal. This implies that in-flow packet order is preserved and that a packet can only be recirculated once. 

Since recirculated packets compete for resources with incoming packets, we would like to minimize the number of recirculated packets.
Fortunately, recirculation happens only for unmatched packets, with a probability of $\frac{1}{carry\_min + 1}$, where $carry\_min$ is the minimal counter value the packet saw in all pipeline stages. Thus, recirculation becomes less frequent as the measurement progresses and the counters grow.  In Section~\ref{sec:recirc} we show that expected number of recirculated packets is asymptotically bounded by the square root of the number of packets.

We can further bound the expected recirculation ratio at the beginning of the execution by initializing all counter registers to a non-zero minimum value. For example, if we initialize all counters to $100$, we also set an upper bound $1\%$ for recirculation probability. In Section~\ref{compromize:initValue} we show that adding an appropriate initial value has a negligible accuracy impact. 
  
\subsection{Approximating the recirculation probabilities}
Recall that the original RAP algorithm admits packets from new flows with probability $\frac{1}{carry\_min+1}$. Intuitively, a flow needs to arrive $carry\_min+1$ times on average to capture a counter with a value of $carry\_min+1$.

It is straightforward to achieve this probability if a random arbitrary-range integer generator is available: we can generate an integer within $[0,carry\_min]$ and check if it's~0. 
However, sometimes we can only obtain random bits from programmable switch's hardware random source, and this effectively limits us to generate random integers within $[0,2^x-1]$ range.  Without the capability to do division or multiplication, we cannot accurately sample with desired probability  $\frac{1}{carry\_min+1}$. As we show in Section~\ref{compromize:restricted}, we can work around this limit without affecting accuracy.

The most simple workaround is to only use probabilities of the form $2^{-x}$. 
Achieving this probability is done by comparing $x$ random bits to zeroes. 
That is, we recirculate unmatched packets with probability $\frac{1}{carry\_min+1}$ rounded to the next smallest
$2^{-x}$. This is a 2-approximation of the desired recirculation probability. The recirculated packet will update the counter to $2^x$.
Rounding is achieved by using a ternary matching over bits of $carry\_min$ variable to find the highest 1 bit. 
The evaluation in Section~\ref{compromize:restricted} shows that this method has a noticeable but acceptable impact on accuracy.

We now introduce a tighter method for approximating the desired recirculation probability. Inspired by floating point arithmetic, we may decompose 
$carry\_min+1=2^y \times T, T\in[8,16)$  
and use a probability of the form  $\frac{1}{2^{y}}\times \frac{1}{ \lfloor T \rfloor}$ to approximate $\frac{1}{carry\_min+1}$ .
We can directly implement the $\frac{1}{2^y}$, while the $\frac{1}{\lfloor T \rfloor}$ is approximated by
randomly generating an integer
between $[0,2^{N}]$ and comparing it against a pre-computed constant $\lfloor \frac{2^{N}}{\lfloor T \rfloor} \rfloor$, via a lookup table.
Further, to avoid non-integer number representation, we always increment the counter value by $1$ upon recirculation.
This achieves a $9/8$-approximation of the desired recirculation probability. 
Our evaluation shows that the accuracy gains are significant. Yet, this method requires additional pipeline~stages.

\begin{algorithm}[t]\small
\For{$i \gets 1$ \textbf{to} $d$} {
  $l_i \gets h_i(iKey)$\;
  {$*$ Hardware stage $i_A$: access \red{$key_i$} register}\;
	  \If{$key_i[l_i] = iKey$}{
			  $matched_i \gets true$\;
	  }
  {$*$ Hardware stage $i_B$: access \blue{$val_i$} register}\;
	  \If{$matched_i$}{
	  	$val_i[l_i] \gets val_i[l_i] + 1$\;
	  }\Else{
	    $oval_i=val_i[l_i] $
	  }
  {$*$ Hardware stage $i_C$: maintain carry minimum}\;
  \If{($\neg matched_i) \wedge (oval_i < carry\_min$)}{
	    $carry\_min \gets oval_i$\;
	    $min\_stage \gets i$
    }
}
\If{$\bigwedge_{i=1}^d (\neg matched_i )$}{
\Comment{$iKey$ not in cache; do Probabilistic Recirculation.}
	$new\_val=2^{\lceil \log_2(carry\_min) \rceil}$\; 
	{
	Generate random integer $R \in [0,new\_val-1]$, by assembling $\lceil \log_2(carry\_min) \rceil$ random bits}\;
	\If{$R=0$}{
		{clone and recirculate packet}\;	
	}
}
\If{ packet is recirculated }{
  $i \gets min\_stage $\;
  $l_i \gets h_i(iKey)$\;
	  $key_i[l_i] \gets iKey$ \Comment{Hardware stage $i_A$: \red{$key_i$} }\;
	  $val_i[l_i] \gets new\_val$ \Comment{Hardware stage $i_B$: \blue{$val_i$}}\;
  {Drop the cloned copy}\;  
  }
\caption{PRECISION heavy hitter algorithm}
\label{algo:Precision}
\end{algorithm}

\subsection{Putting all adaptations together}
With all the aforementioned hardware-friendly adaptations in mind, we assemble the PRECISION algorithm, which satisfies all hardware-imposed constraints of the RMT architecture. 
Algorithm~\ref{algo:Precision} is a pseudocode version of PRECISION. 
Line 1 reflects PRECISION's $d$-way associative memory access, iterating through each way.
In Line 8 we increment the counter for matched packets, while unmatched packets handled between Line 15 and Line 19. 
We flip a coin in Line 17, and the 2-approximation of recirculation probability manifests in Line 16.  
Recirculated packets update register memory corresponding to their minimal entries. This is described between Line 20 to Line 24. 
We highlighted accesses to register memory in color, note that registers are only accessed once per stage. Each branching fits in a transition between hardware pipeline stages, removing the need to perform in-stage branching.  

\subsection{Parallelizing actions to reduce hardware stages used}
\label{sec:parallel}

Algorithm~\ref{algo:Precision} presented PRECISION in its most straightforward arrangement, iterating through the $d$-way in tandem, while each uses three pipeline stages. This costs as much as $d\times 3$ hardware pipeline stages for register memory reads. Since the total number of stages is very limited, we explain how to optimize the required number of stages further,
and fit a larger $d$ on the same hardware.
This optimization may also be applicable to other algorithms with a similar repeated register array access pattern.

Intuitively, each `if' in the pseudocode is a branching, separating the algorithm into different hardware stages. However, it may be possible to group independent stages and reduce the total number of hardware stages needed. 



In our implementation, PRECISION requires two branching for each of the $d$ ways. That is, it requires three pipeline stages for each way. 
The stages in each way are:\\ 
\textbf{Stage A:} Read flow ID from flow entry array.
{\textit{(branching: does entry's ID match my ID?)}}\\
\textbf{Stage B:} Read/Update 
from the counter array.
{\textit{(branching: is counter smaller than the current  minimum?)}}\\
\textbf{Stage C:} Compute and ``carry'' the new minimum~value.

If we indeed require three hardware stages for each pair of flow entry array and counter array, a switch with $X$ physical stages can at most implement PRECISION with $d=X/3$. This assumes that all pipeline stages serve for heavy-hitter detection. In practice, we would like to leave enough pipeline stages for other network applications.

\begin{figure}[t]
\centering
\includegraphics[width=0.5\textwidth]{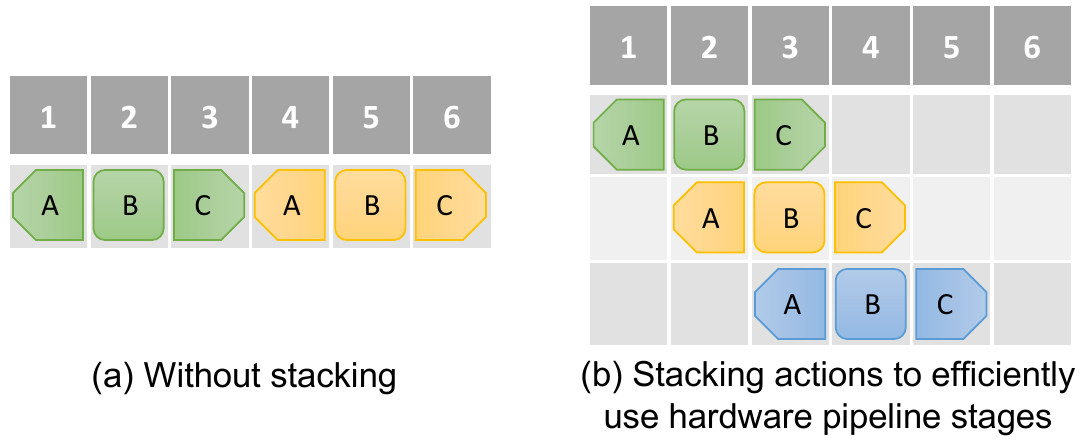}
\caption{We reduce the number of pipeline stages used by stacking together independent actions between different ways. For $d$-way PRECISION, this reduces the number of pipeline stages required from $d\times 3$ to $d+2$. 
}
\label{fig:illustration2}
\ifdefined\submissionVersion
	\vskip -0.5cm
\fi
\end{figure}

However, our algorithm does not have a hard dependency between different groups of stages. If we denote the $d$ ways as $1$, $2$, $3$ and the three pipeline stages for each action as $A$, $B$, and $C$, we can observe that (for example) $2_A$ and $1_C$ are independent. 
Thus, it's not necessary to serialize everything into the pattern shown in Figure~2(a). Instead, we can ``stack" operations from different groups together, as shown in Figure~2(b). Specifically,  reading the flow identifier for the next flow entry array can be parallelized with incrementing a counter for the previous way's counter array and so forth. 
Therefore, we can parallelize different execution stages of multiple ways as there is no direct causal relation or data dependency between stage action $(i+1)_B$ and $i_C$, or between $(i+1)_A$ and $i_B$. Thus, by using the stacking pattern shown in Figure~2(b), we reduce the number of required stages to implement $d$-way PRECISION from $d\times 3$ to $d+2$, amortizing to one stage per way. 
\footnote{There is indeed a causal dependency between stage $(i+1)_C$ and $i_C$ when computing the carried minimum value $carry\_min$, thus using only a constant number of 3 hardware stages is not possible.
Also, other hardware constraints that limit the number of parallel actions in one hardware stage exists, but these are less stringent than the limit on the total number of hardware stages.}

For a programmable switch with a limit of $X$ hardware stages, the actual maximum $d$ we can implement will be smaller, because we need extra stages before and after the core algorithm for setup and teardown, such as extracting flow ID and performing random coin-tossing. Furthermore, a network switch will need to fulfill its regular duties like routing, access control, etc., and would not devote all its resources to the PRECISION algorithm. Nevertheless, we can expect any commodity programmable switch to run the $d=2$ version of PRECISION smoothly, alongside its regular duties. When extra resources are available, we may increase $d$ to improve accuracy as shown in Section~\ref{compromize:associativy}. 


\section{Bounding the Amount of Recirculation}
\label{sec:recirc}
Here we show a bound on the total number of packet-recirculations. Our main result, Theorem~\ref{thm1}, shows that the number of recirculated packets is sublinear. Combined with our approach for setting initial values to counters to avoid high recirculation ratio at the beginning, we maintain recirculation at acceptable levels throughout the measurement. 
\ifdefined\submissionVersion
We start with a lemma that explains how recirculation behaves with a single counter. 
The proof is deferred to the full version~\cite{fullVersion}.
\else
We first present an auxiliary lemma about summing random variables.
The proof is deferred to the appendix.
\fi
\ifdefined\submissionVersion
\else
\begin{lemma}\label{lem:iidGeometric}
Fix some $p\in(0,1]$, $T\in\mathbb N^+$ 
and let $X_1,X_2,\ldots\sim Geo(p)$ be independent geometric random variables with mean $1/p$. Denote by $Z\triangleq \min\{{n\in\mathbb N}\mid \sum\limits_{i=1}^nX_i\ge T\}$ the minimal number $n$ such that the sum of $X_1,\ldots,X_n$ exceeds the threshold $T$. Then $\mathbb E[Z] = p(T-1)+1$.
\end{lemma}
\fi
\ifdefined\submissionVersion
\else
Next, we show a bound on the expected number of packets that would be sampled by a single-counter PRECISION instance. For this, we denote by $X_i$ the number of packets between the time that the counter has reached a value of $i$ and the time it first reaches of $i+1$. The proof is delayed to the appendix.
\fi
\ifdefined\submissionVersion
\begin{lemma}\label{lem:singleCounter}
Fix some $T\in\mathbb N^+$ and let $\set{X_i\sim Geo(1/i)\mid i\in\mathbb N}$ denote independent geometric random variables such that $\mathbb E[X_i]=i$. Let $A\triangleq \min\{{n\in\mathbb N}\mid \sum\limits_{i=1}^n X_i\ge T\}$. Then $\mathbb E[A] \le 2\sqrt{T}$.
\end{lemma}
\else
\begin{lemma}\label{lem:singleCounter}
Fix some $T\in\mathbb N^+$ and let $\set{X_i\sim Geo(1/i)\mid i\in\mathbb N}$ denote independent geometric variables such that the expectation of $X_i$ is $i$. Similarly to the above lemma, let $A\triangleq \min\{{n\in\mathbb N}\mid \sum\limits_{i=1}^n X_i\ge T\}$ denote the number of variables needed to cross the threshold $T$. Then $\mathbb E[A] \le 2\sqrt{T}$.
\end{lemma}
\fi
We now present the main theorem.
Note that here we assume ideal random recirculation probability $1/i$, and the approximation techniques only reduce recirculation~further.
\begin{theorem}
\label{thm1}
Denote the number of packets in the stream by $N$ and the number of counters by $C$. The expected number of recirculated packets is $O(\sqrt {NC})$. 
\end{theorem}
\begin{proof}
For $i\in\set{1,\ldots,C}$, let $R_i$ denote the number of times PRECISION recirculates a packet to update the $i$'th counter and by $R$ the overall recirculation. Next, let $N_i$ denote the number of times this counter was probabilistically modified (that is, a packet traversed the entire pipeline, and this counter was the minimal along its path). We have that $\sum_{i=1}^C N_i\le N$ (this is inequality as some packets update their flow counter and are surely not recirculated). According to Lemma~\ref{lem:singleCounter} we have that $\mathbb E[R_i]\le 2\sqrt{N_i}$. 
This gives
\ifdefined\submissionVersion
$\mathbb E[R] 
= \sum_{i=1}^C \mathbb E[R_i] 
\le \sum_{i=1}^C 2\sqrt{N_i}
\le 2\sqrt{\sum_{i=1}^C N_i}
= 2\sqrt{NC},$
\else
$$\mathbb E[R] 
= \sum_{i=1}^C \mathbb E[R_i] 
\le \sum_{i=1}^C 2\sqrt{N_i}
\le 2\sqrt{\sum_{i=1}^C N_i}
= 2\sqrt{NC},$$
\fi
where the last inequality follows from the concaveness of the~square~root.
\end{proof}

\section{Evaluation}
This section presents an evaluation of PRECISION's accuracy and adaptation mechanisms.
We verified PRECISION using Barefoot's Tofino emulator; however, due to performance reasons, we could not use it for the actual evaluation.
Instead, we use Python to implement various measurement algorithms and compare their accuracy. Python-based emulation also allows us to manipulate hardware parameters freely, so we can independently manipulate each hardware restriction.
We start by studying the effect of each hardware restriction on PRECISION's accuracy. 
Next, we compare PRECISION to related work, including HashPipe~\cite{HashPipe},  as well as Space-Saving~\cite{SpaceSavings} and RAP~\cite{RAP} that are not directly implementable on programmable switches.

Our evaluation utilizes the following datasets:  
\\
\textbf{CAIDA}: The CAIDA Anonymized Internet Trace 2016 \cite{CAIDA} (in short, \emph{CAIDA}). Data is collected from the `equinix-chicago' Internet backbone link and contains a mix of UDP, TCP, and ICMP packets. We used packets' Source-Destination IP address pair as their flow ID.\\
\textbf{UWISC-DC}: A data center measurement trace recorded at the University of Wisconsin~\cite{IMC10trace}. 
\\
\textbf{UCLA}: The University of California, Los Angeles Computer Science department packet trace (denoted \emph{UCLA})\cite{UCLA}. 
We also tested our algorithm using synthetic trace with Zipf distribution and observed similar results.
	
All experiments were performed with 2 million packets using a software emulated version of PRECISION, \mbox{and we repeated each experiment 10 times.}

\begin{figure}[t]
\centering
\begin{tabular}{cc}\subfloat[\label{fig:assoctive} Frequency Estimation]{\includegraphics[width = 0.45\columnwidth]
			{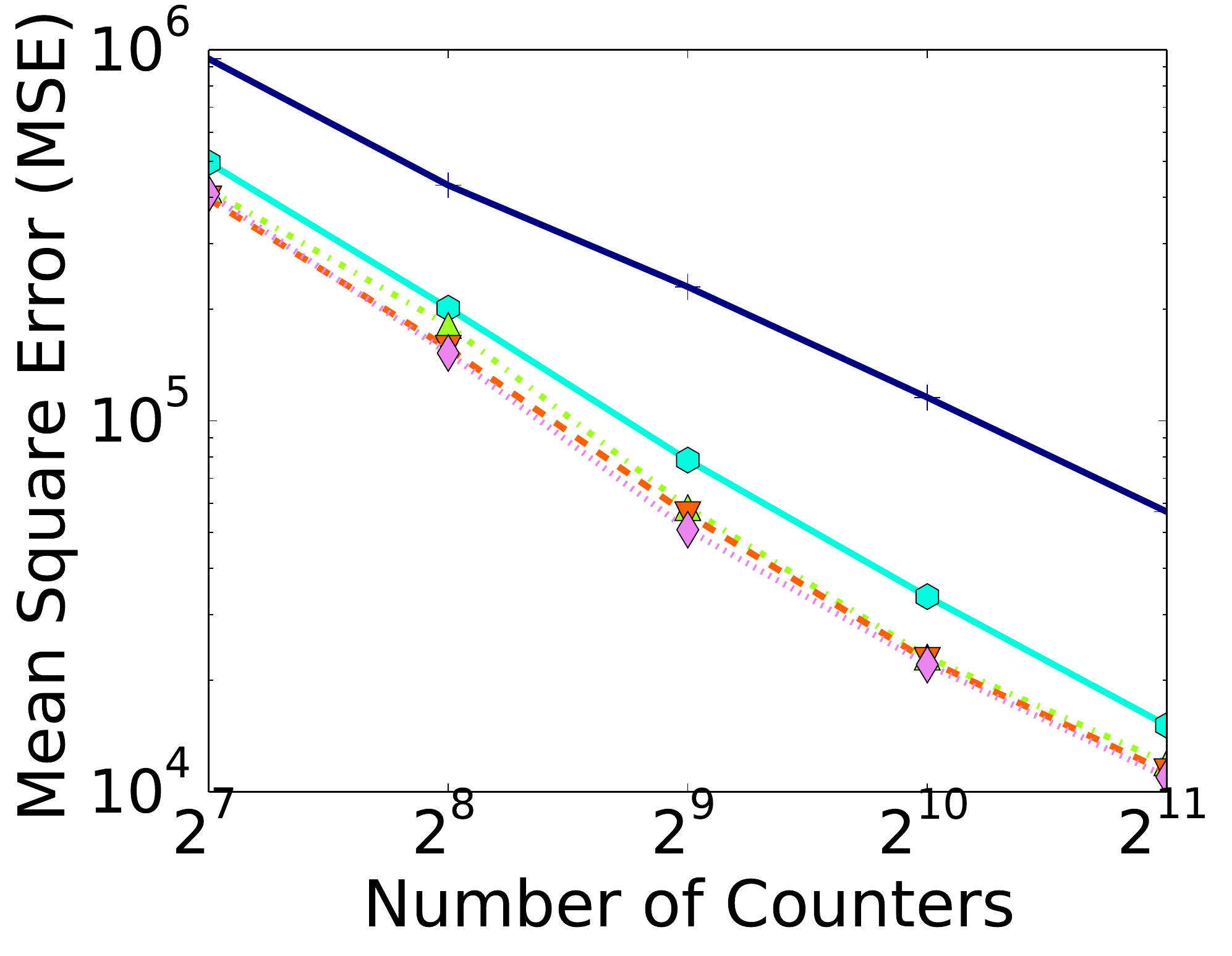}} &    
		\subfloat[\label{fig:assocTopK} Top-128]{\includegraphics[width = 0.45\columnwidth]
			{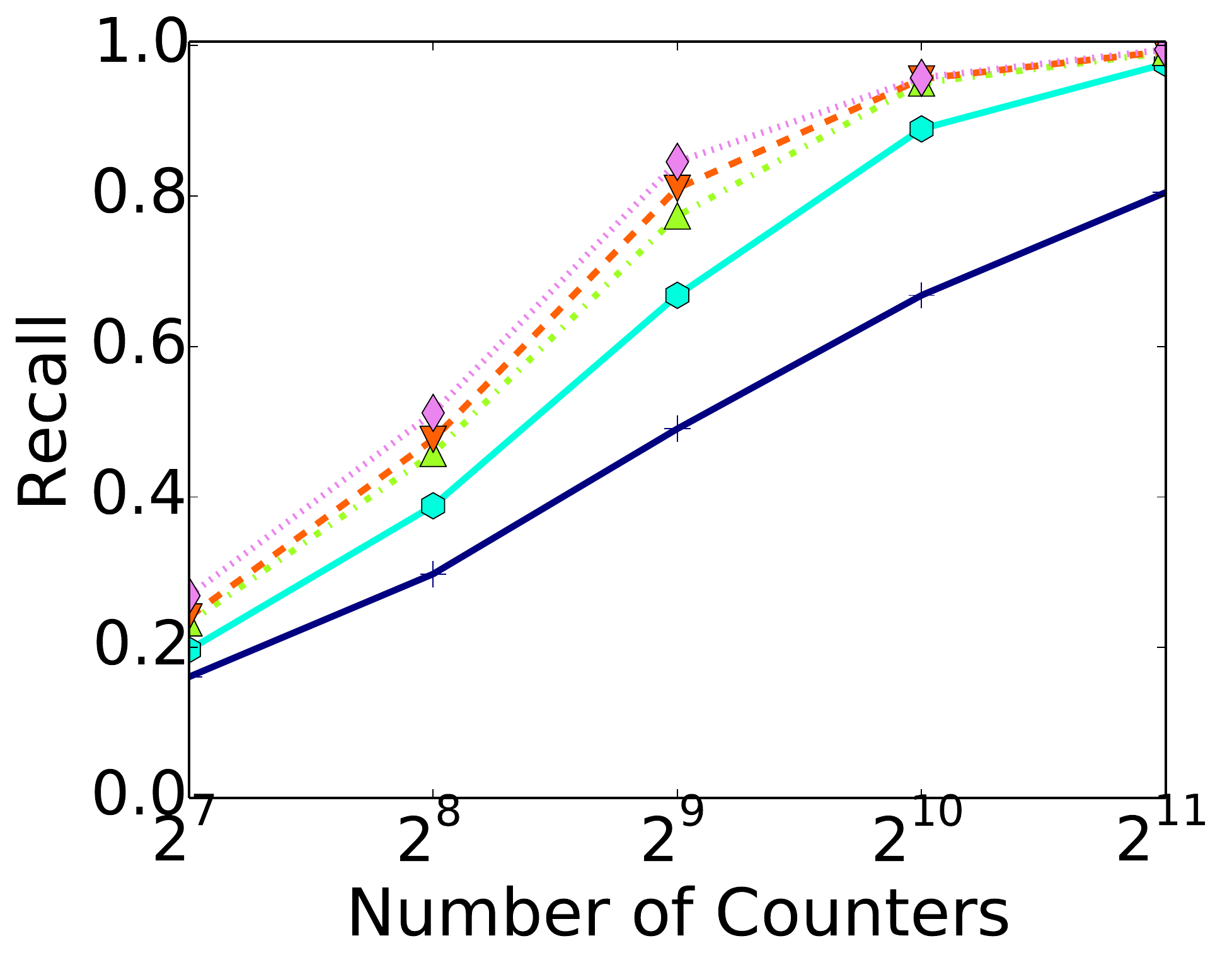}}  
	\end{tabular} 
\includegraphics[height = 0.3cm]
		{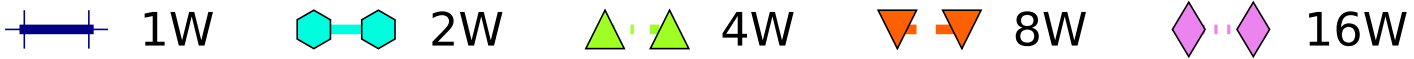}
	\ifdefined\submissionVersion
	\vskip -0.1cm
	\fi
    \caption{\label{fig:assoc_combined} Effect of limited associativity on the frequency estimation error and top-$k$ recall, on CAIDA~trace. Using $d=2$-way is a right balance between achieving good accuracy and saving pipeline stages usage.
    }
\end{figure}

\subsection{Limited associativity}
\label{compromize:associativy}
We start with the frequency estimation problem and measure OnArrival error. In this measurement, we evaluate PRECISION with a varying number of ways ($d$) and use the same amount of total memory for all trials. Our results in Figure~\ref{fig:assoctive} show that for this problem 1-way associativity ($d=1$) is a bit too low, but 2-way is already reasonable and further increasing $d$ has diminishing returns. 
Figure~\ref{fig:assocTopK} evaluates how $d$ affects the Recall in top-$k$ problem, using 512 counters to find top-128 flows. In this metric, we see that associativity is more important than in frequency estimation.  $d=2$ requires up to 2$\times$ more counters than $d=16$ to achieve the same recall. Changing to smaller or larger $k$ yields similar observation.

We conclude that limited associativity incurs minimal accuracy loss in frequency estimation and is more noticeable in top-$k$. Our suggestion is to use $d=2$ as it achieves the right balance between accuracy and the number of pipeline stages. 

\begin{figure}[t]
\vskip -0.5cm
\centering
\begin{tabular}{cc}\subfloat[\label{fig:packetFE} Frequency Estimation]{\includegraphics[width = 0.45\columnwidth]
			{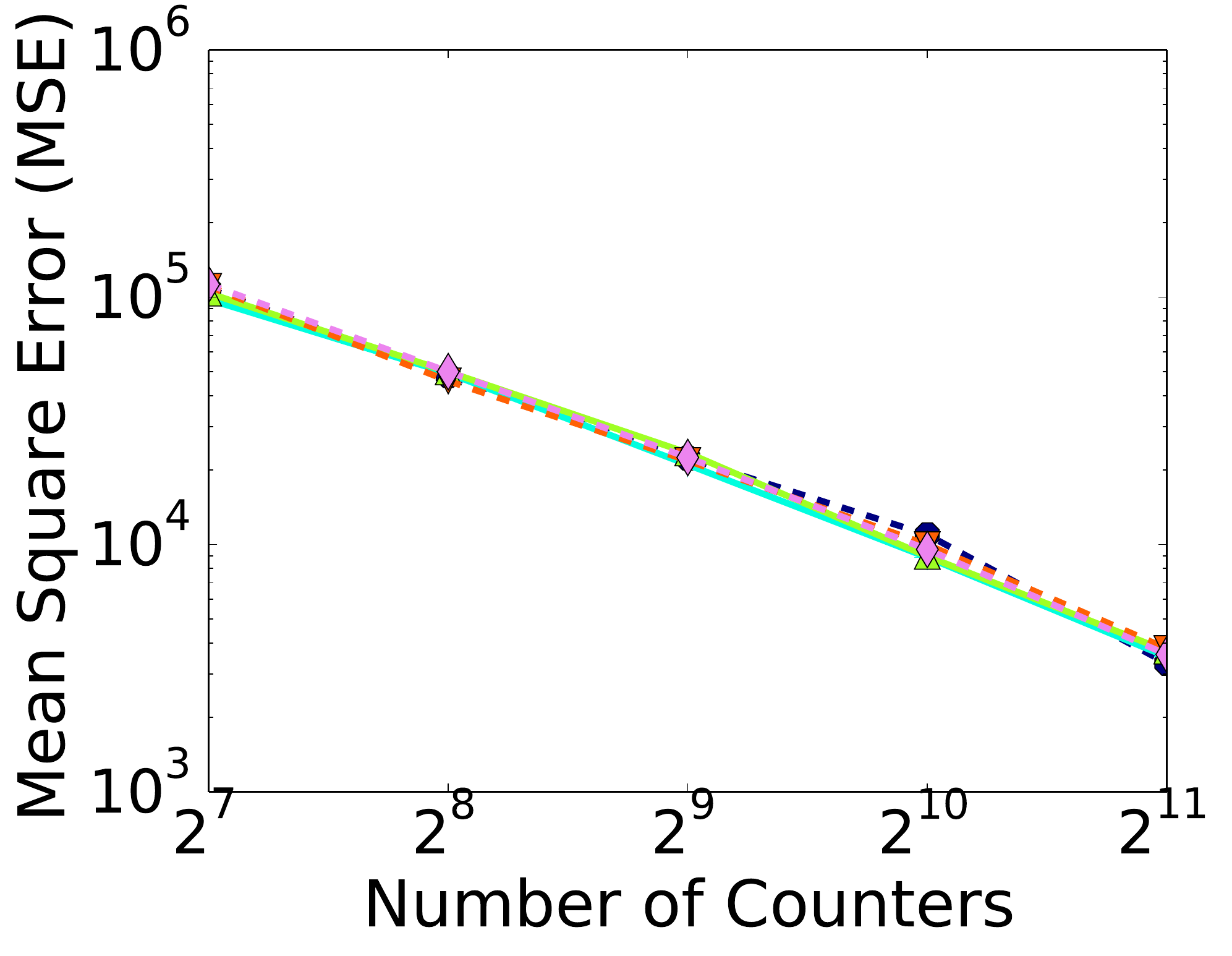}} &    
		\subfloat[\label{fig:packetTopK} Top-128]{\includegraphics[width = 0.45\columnwidth]
			{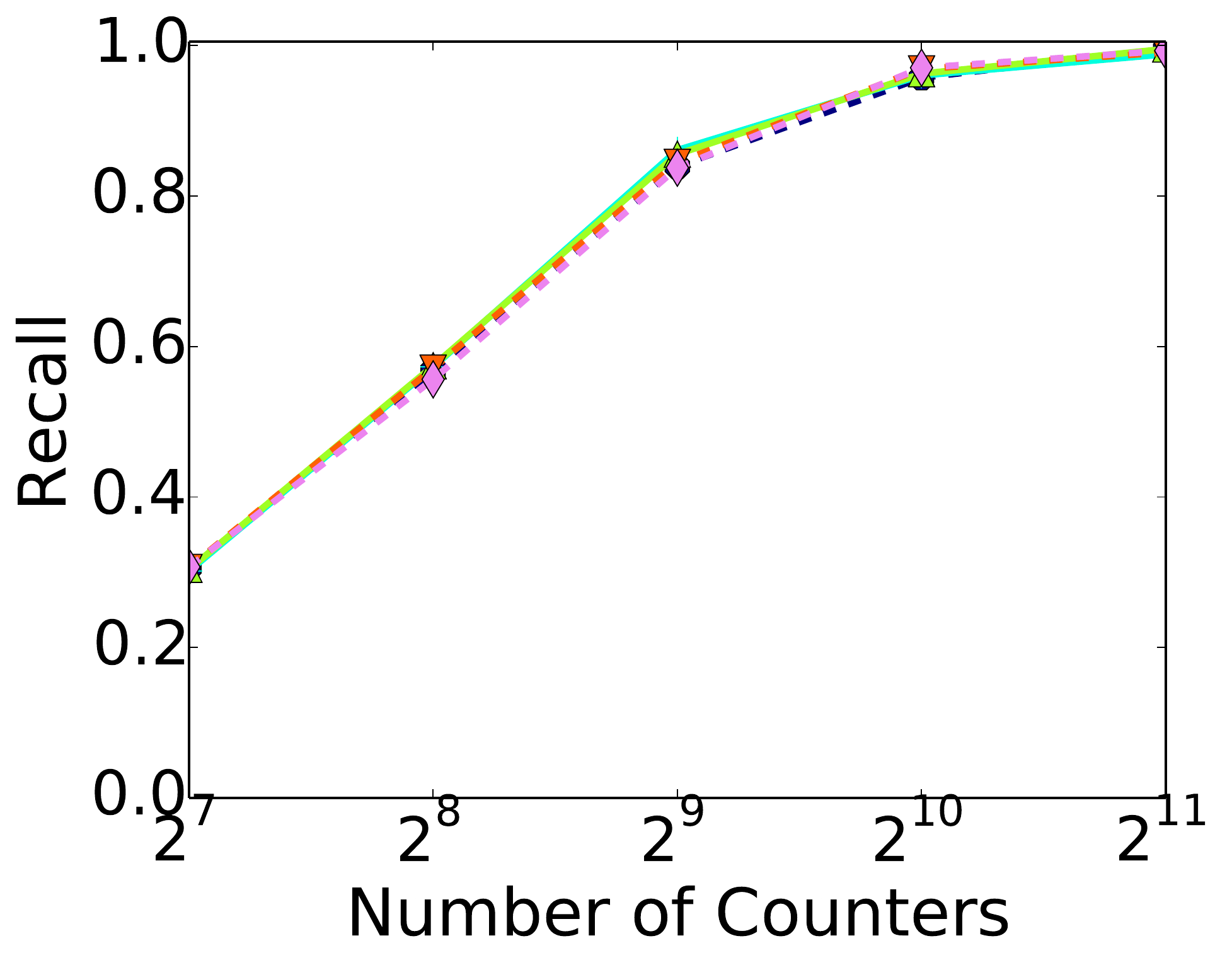}}  
	\end{tabular} 
\includegraphics[width=\columnwidth]
		{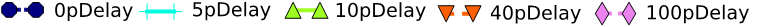}
		\ifdefined\submissionVersion
		\vskip -0.2cm
		\fi
    \caption{\label{fig:delay_combined} Effect of the delayed update on the frequency estimation error and top-$k$ Recall, on CAIDA~trace.  Even a delay of 100 packets has minimal impact on the~accuracy. }
    \ifdefined\submissionVersion
		\vskip -0.5cm
	\fi
\end{figure}
\begin{figure*}[]
\centering
\begin{tabular}{ccc}\subfloat[\label{fig:initFE} Frequency Estimation]{\includegraphics[width = 0.66\columnwidth]
			{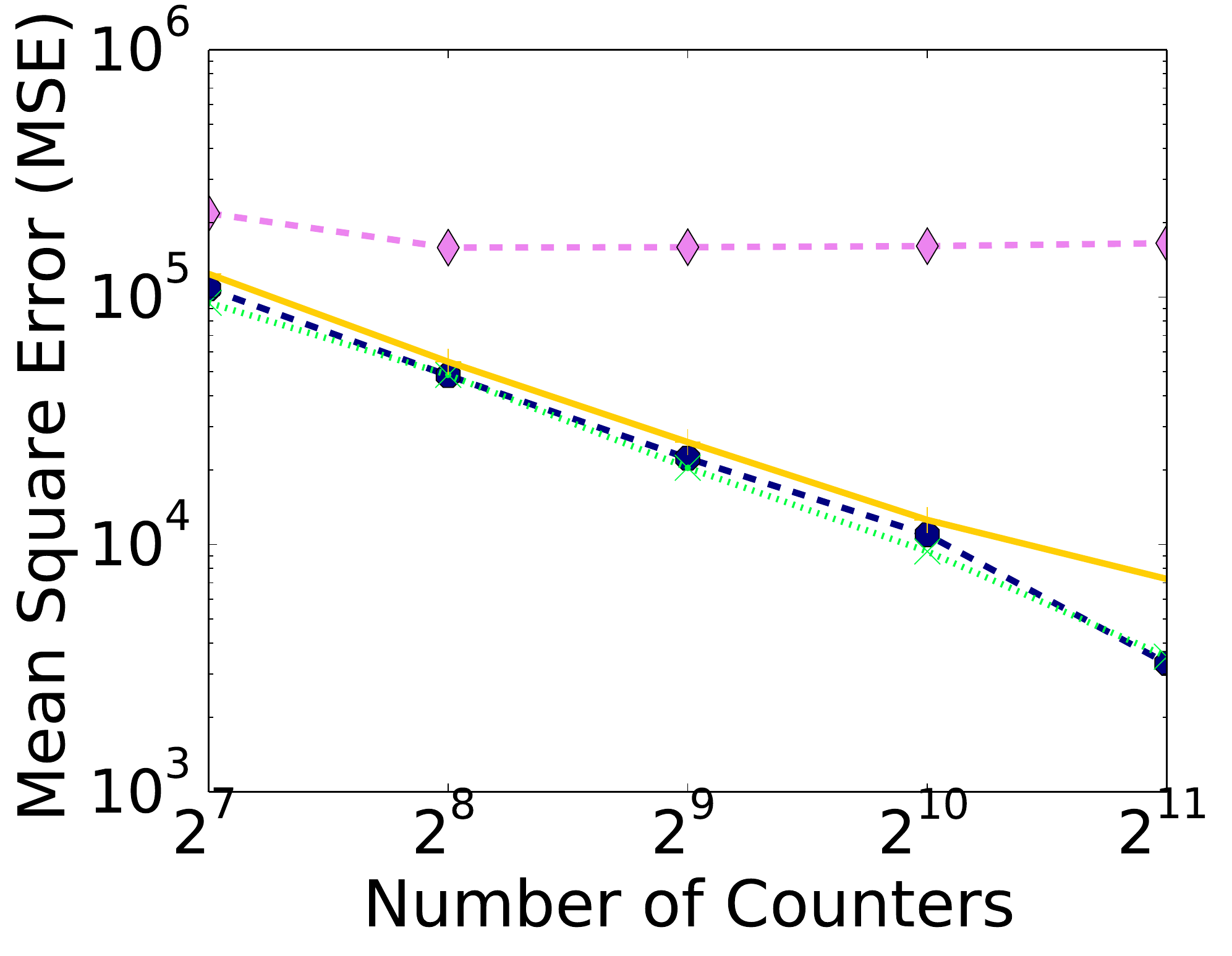}} &  
		\subfloat[\label{fig:initTopK} Top-128]{\includegraphics[width = 0.66\columnwidth]
			{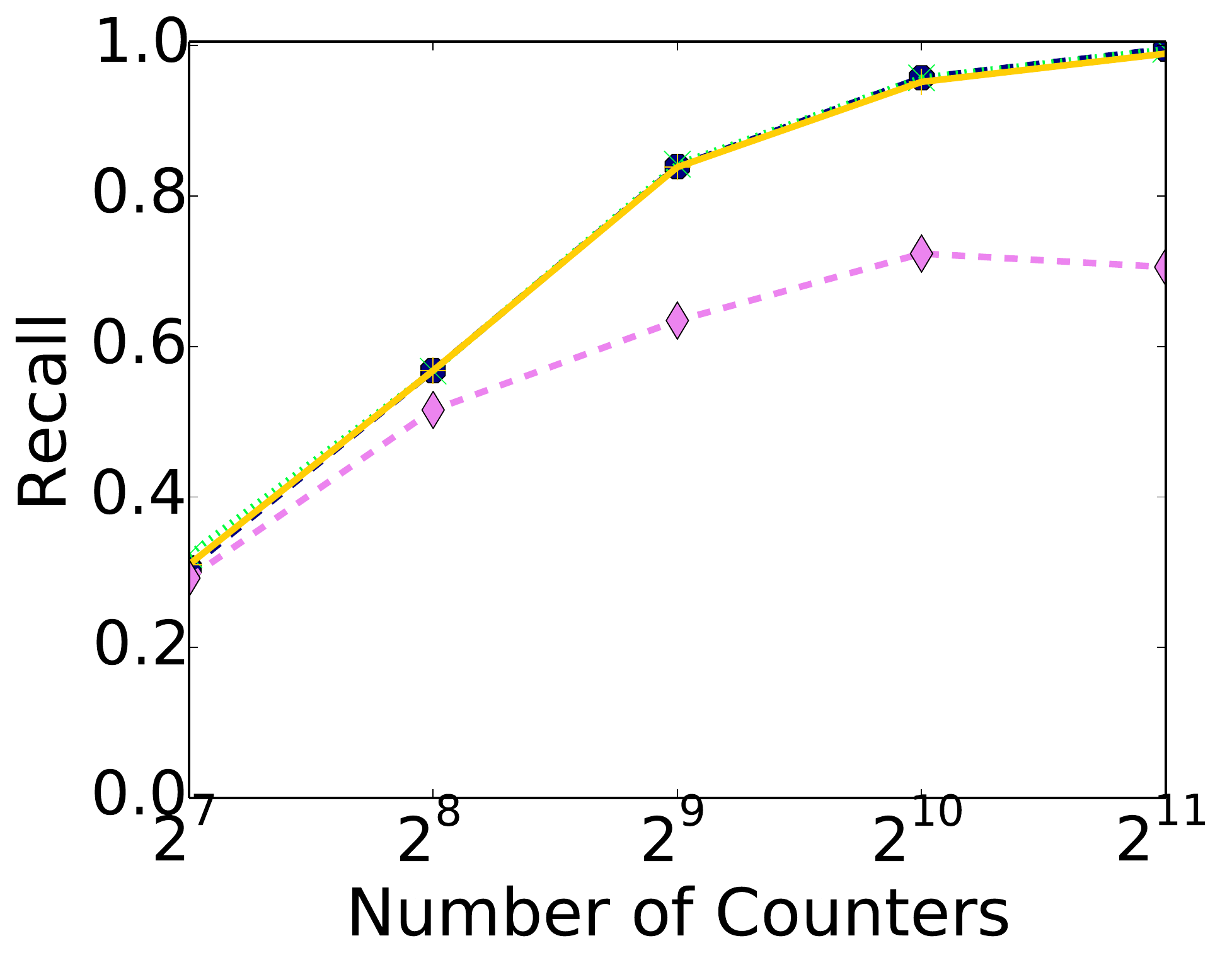}}  &
         \subfloat[\label{fig:InitConverge} Convergence of Top-128 Recall over time]{\includegraphics[width = 0.66\columnwidth]
			{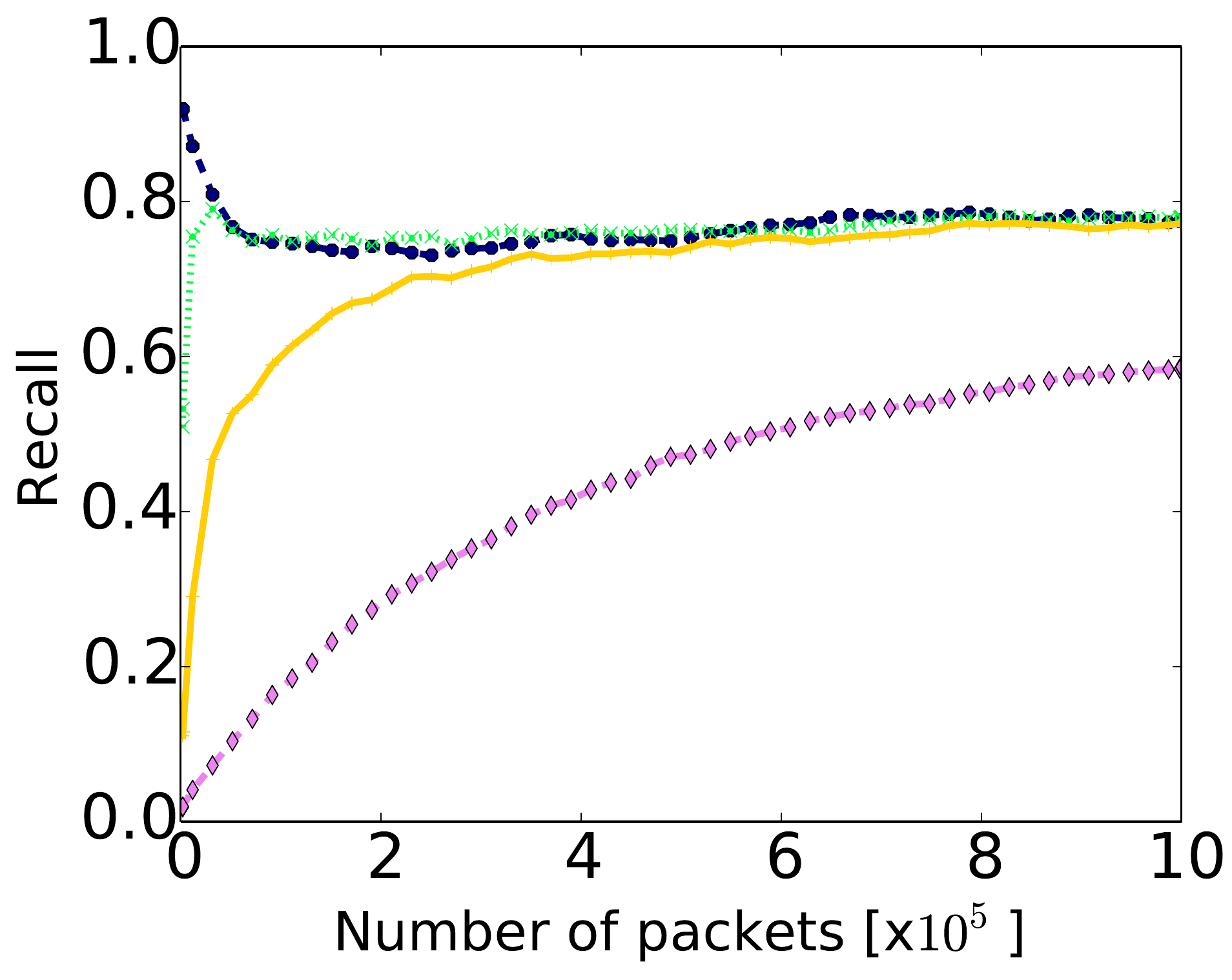}}  
	\end{tabular} 
\includegraphics[height = 0.38cm]
		{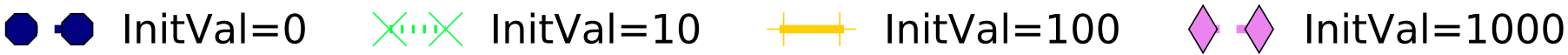}
		\ifdefined\submissionVersion
		\vskip -0.1cm
		\fi
    \caption{\label{fig:init_combined} Effect of initial value on the overall frequency estimation error and top-$k$ recall, on CAIDA trace. 
    An initial value of 100  leads to fast convergence and does not hurt accuracy, while upper-bounding recirculation to $1\%$.
    }		\ifdefined\submissionVersion
		\vskip -0.5cm
		\fi
\end{figure*}

\ifdefined\submissionVersion
		\vskip -0.2cm
\fi
\subsection{Entry update delay}
\label{compromize:packetDelay}
We now evaluate the impact of update delay between the decision to recirculate and the actual flow entry update.
Instead of using empirical evidence on one particular programmable switch, we simulate various possible delay values in terms of pipeline length.  
Figure~\ref{fig:packetFE} shows results for the MSE (Mean Square Error) in the frequency estimation problem and Figure~\ref{fig:packetTopK} shows the Recall in top-$k$ problem when trying to find the top-128 flows. As can be observed, the lines are almost indistinguishable. That is, update delay has a minor impact on accuracy for both metrics, even for a delay of 100 packets.  We assume that practical switching pipelines would have shorter recirculation delays, as today's programmable switches have much less than 100 stages. A possible reason for this insensitivity to update delays is that replacing flow entries is already a rare and random event. Thus, the actual replacement time barely affects the accuracy even if it slightly deviates from the decision time.

\subsection{Initial value}
\label{compromize:initValue}
We now evaluate the impact of having an initial value larger than zero set to all counters. Intuitively,  the initial value limits the number of recirculated packets, but also requires some time to converge. 
This is because having a non-zero initial value means that we need to see more unmatched packets before we claim an entry --- even if that entry is empty.  
Figure~\ref{fig:initFE} show results for the frequency estimation metric. As can be observed, the initial value does affect the accuracy, and the effect is small until initial value 100, but initial value 1,000 causes a large impact.
A similar picture can be observed in Figure~\ref{fig:initTopK} that evaluates Recall in the top-$128$ problem using 512 counters. As depicted,  initial value also has a little impact up to 100, but an initial value of 1,000 results in a poor Recall.

Figure~\ref{fig:InitConverge} completes the picture by showing the change of the Recall over time when trying to find top-$128$. As shown, the convergence time is inversely correlated with the initial value. In most cases, 1 million packets are enough for converging with an initial value of 100. We observed similar behavior for different packet traces. It appears that an initial value of 1,000 requires more packets to converge.

We conclude that a small initial value has a limited impact on the performance when the measurement is long enough.
To facilitate quick convergence, we suggest an initial value of 100, as it seems reasonable to upper bound recirculation to at most 1\% of the packets, and the convergence time is shorter than 1 million packets, which translates to less than one second on fully-loaded 10~Gbps links. 
\begin{figure*}[h]
  \vskip -0.5cm
	\begin{tabular}{cccc}
		\subfloat[CAIDA]{\includegraphics[width = 0.23\linewidth]
			{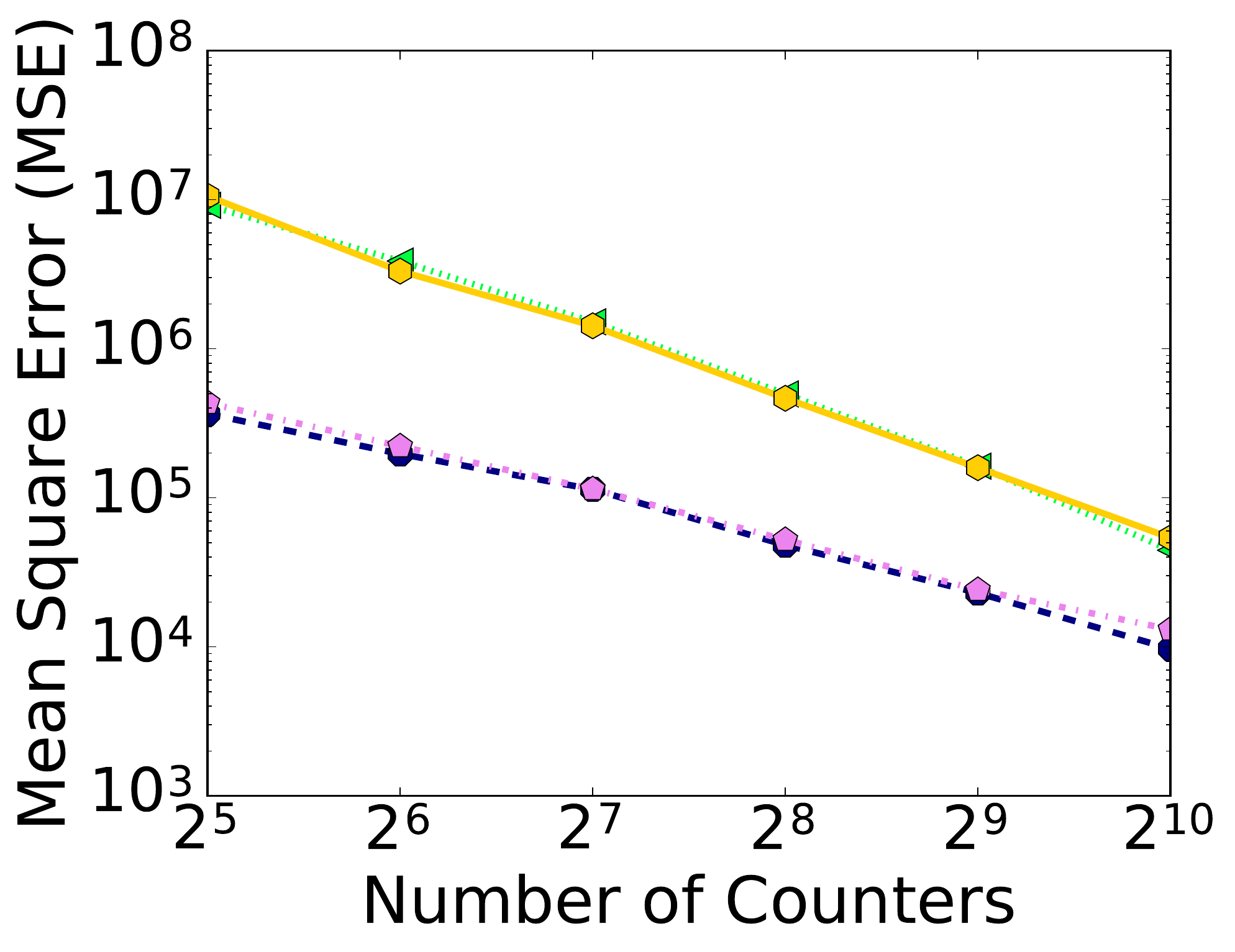}} &    
		\subfloat[UWISC-DC]{\includegraphics[width = 0.23\linewidth]
			{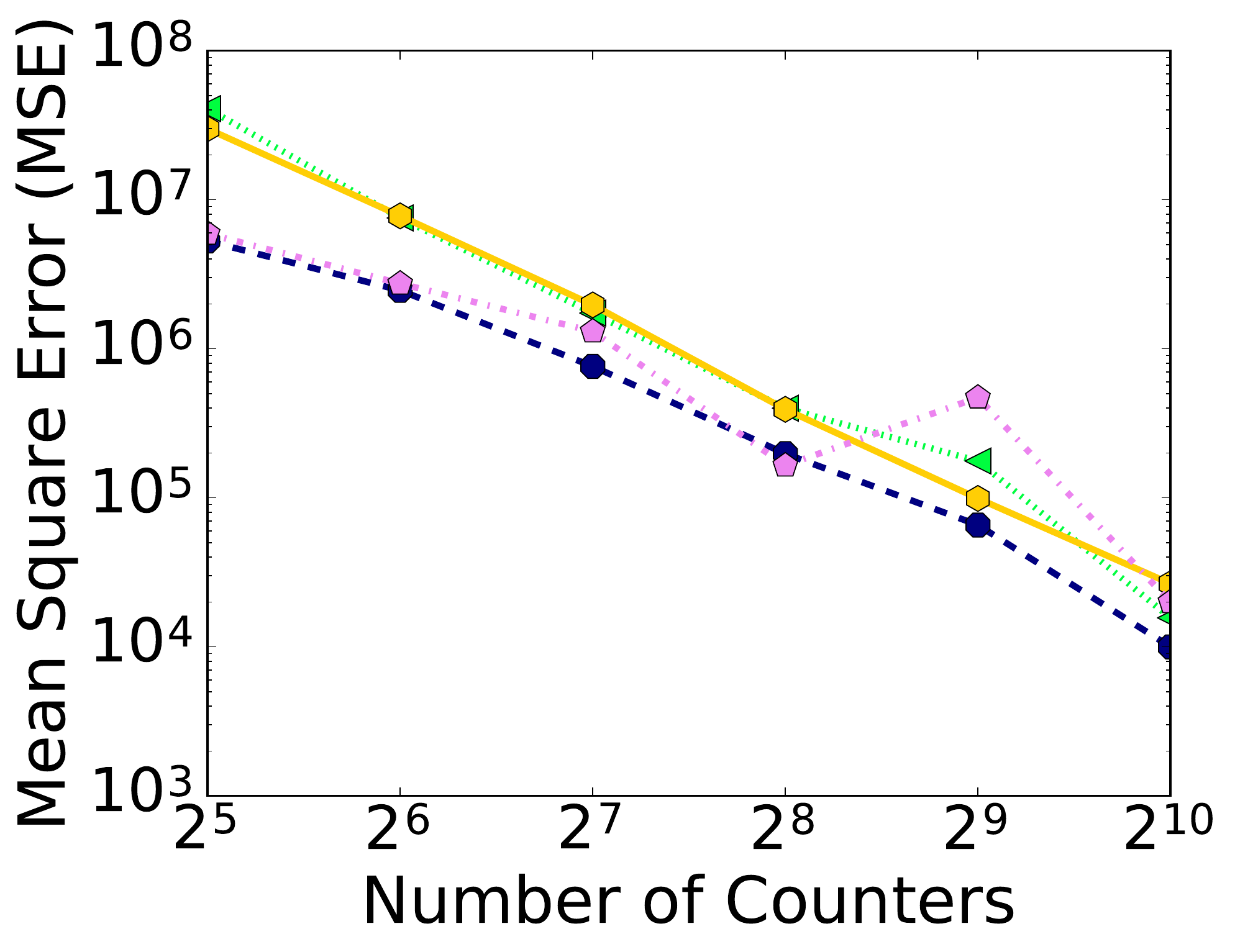}} &    			
		\subfloat[CAIDA]{\includegraphics[width = 0.23\linewidth]
			{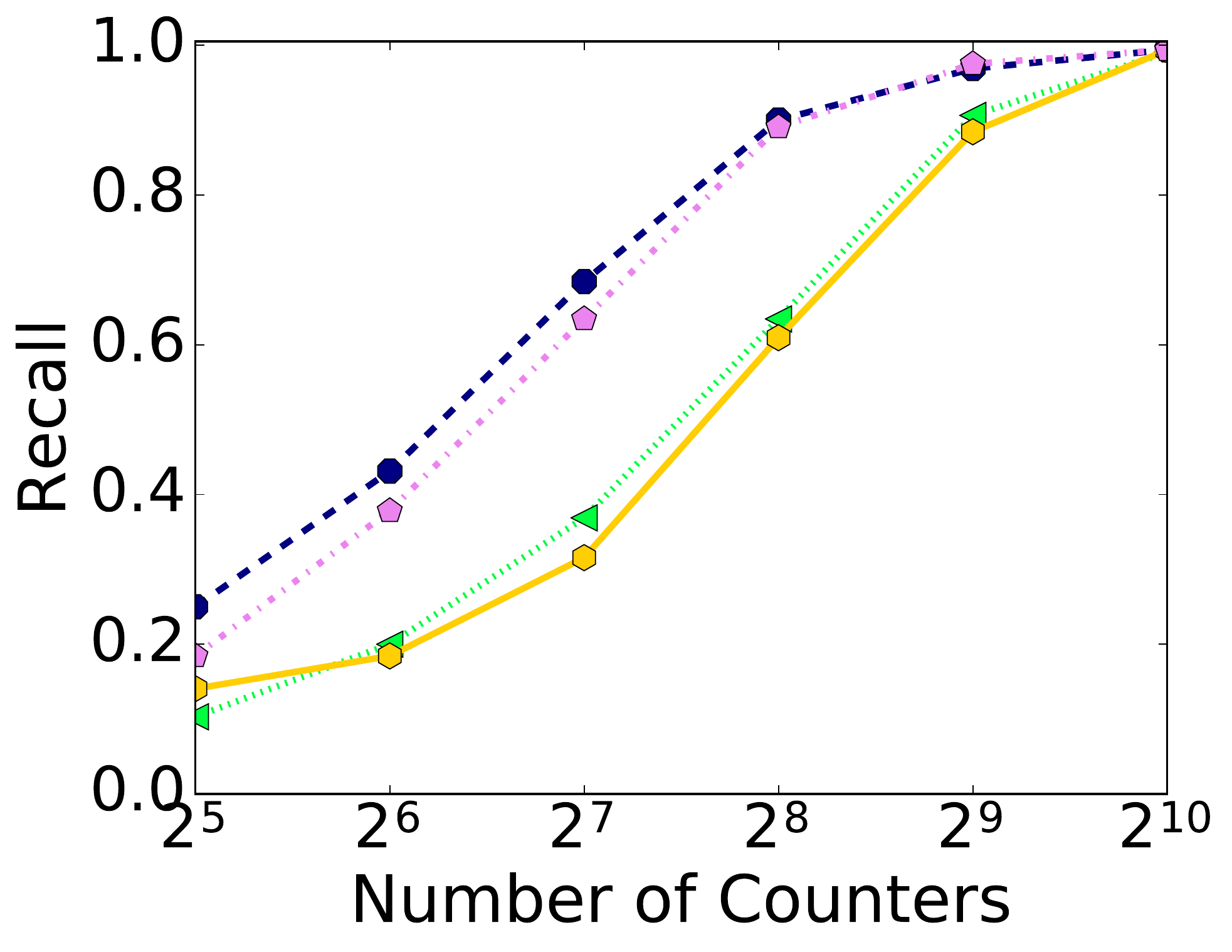}} & 
		\subfloat[UWISC-DC]{\includegraphics[width = 0.23\linewidth]
			{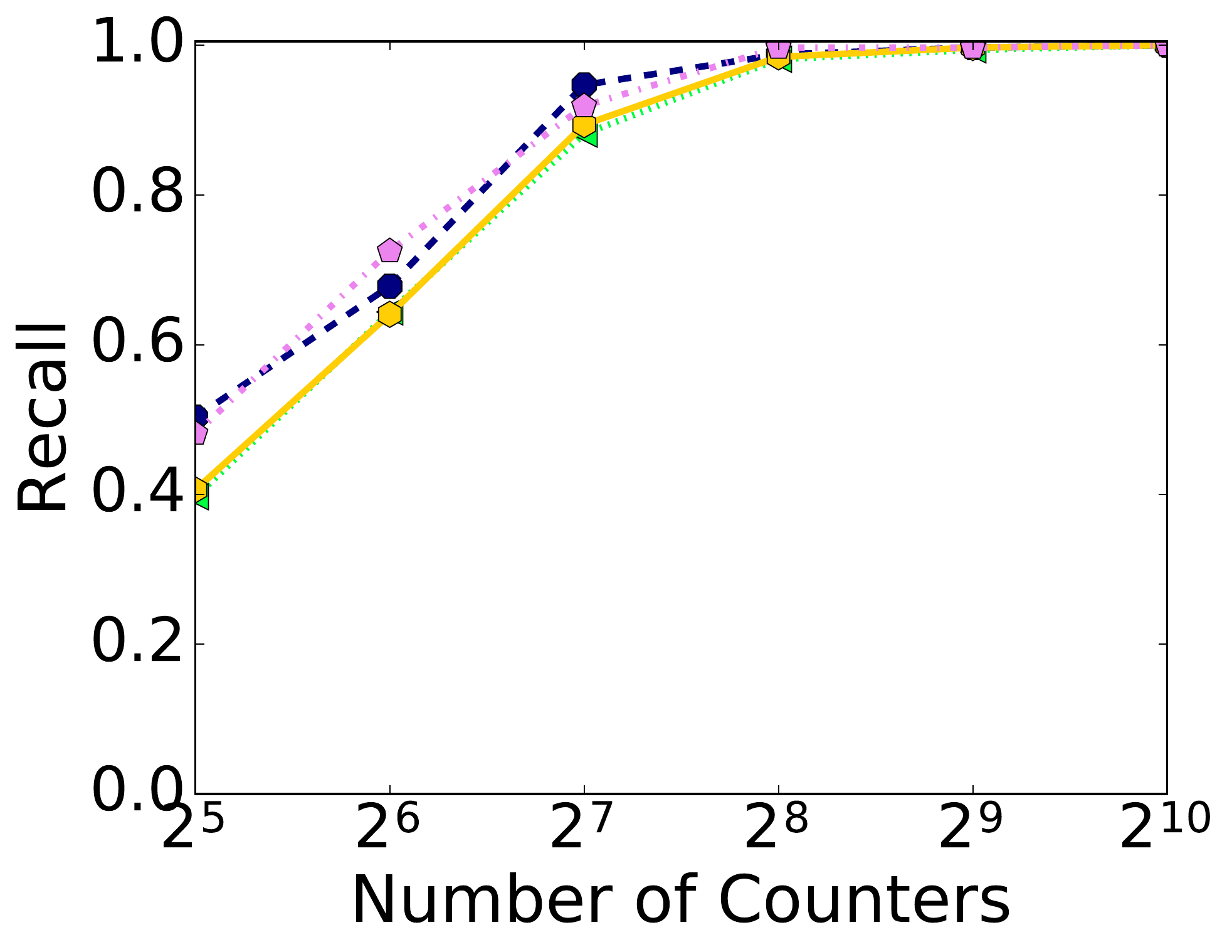}}    
	\end{tabular}    
	\centering{\includegraphics[height = 0.3cm,trim=6 8 6 8, clip]
		{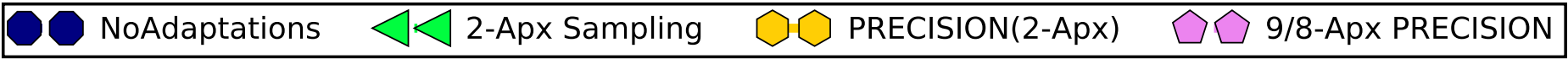}} 
		\ifdefined\submissionVersion
		\vskip -0.1cm
		\fi
		\caption{\label{fig:RestrictFE}The effect of approximating the recirculation probabilities on the accuracy for frequency estimation and top-32.
		}\vskip -0.3cm
\end{figure*}
\begin{figure*}[]
	\begin{tabular}{ccc}
		\subfloat[CAIDA]{\includegraphics[width = 0.66\columnwidth]
			{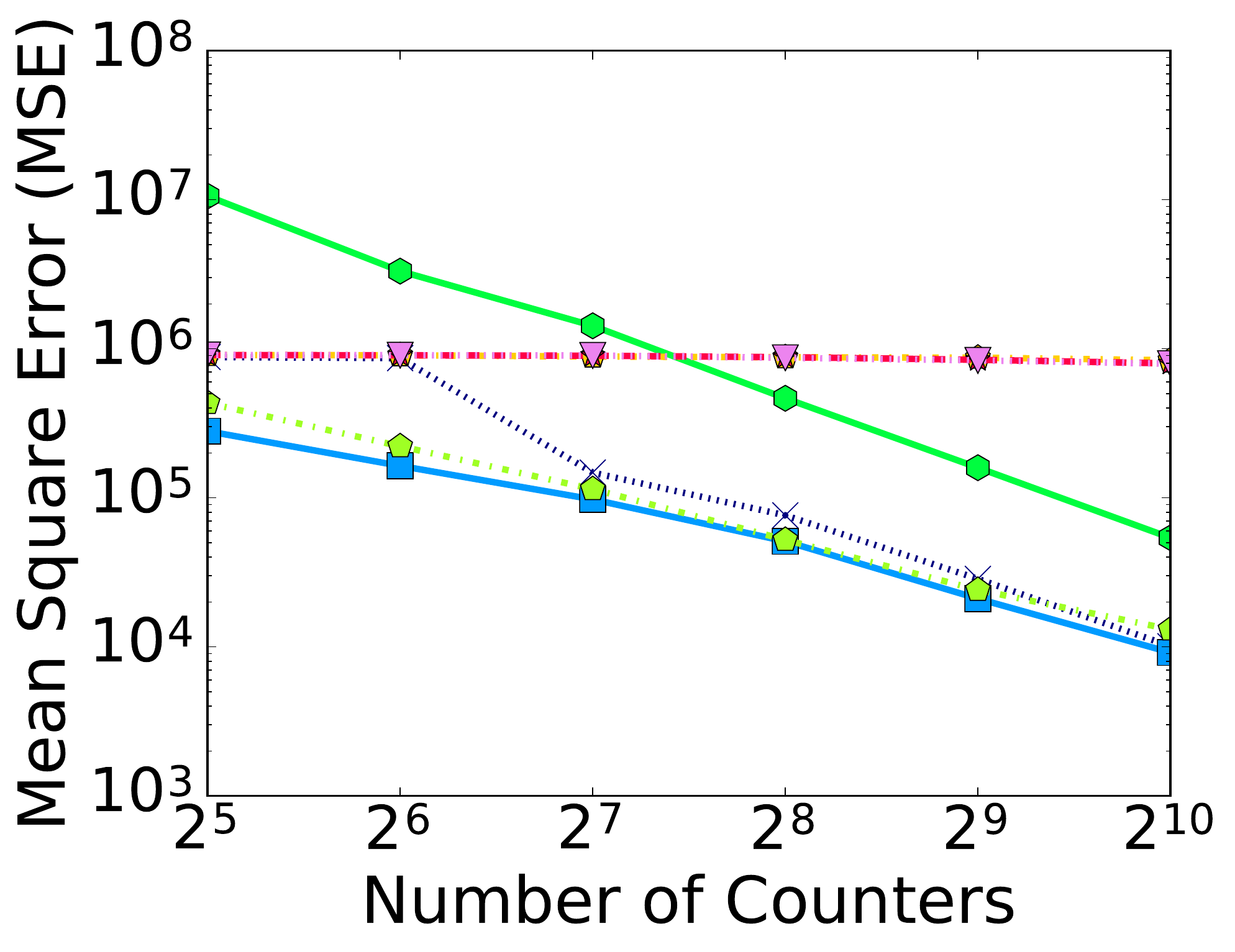}} &    
		\subfloat[UCLA]{\includegraphics[width = 0.66\columnwidth]
			{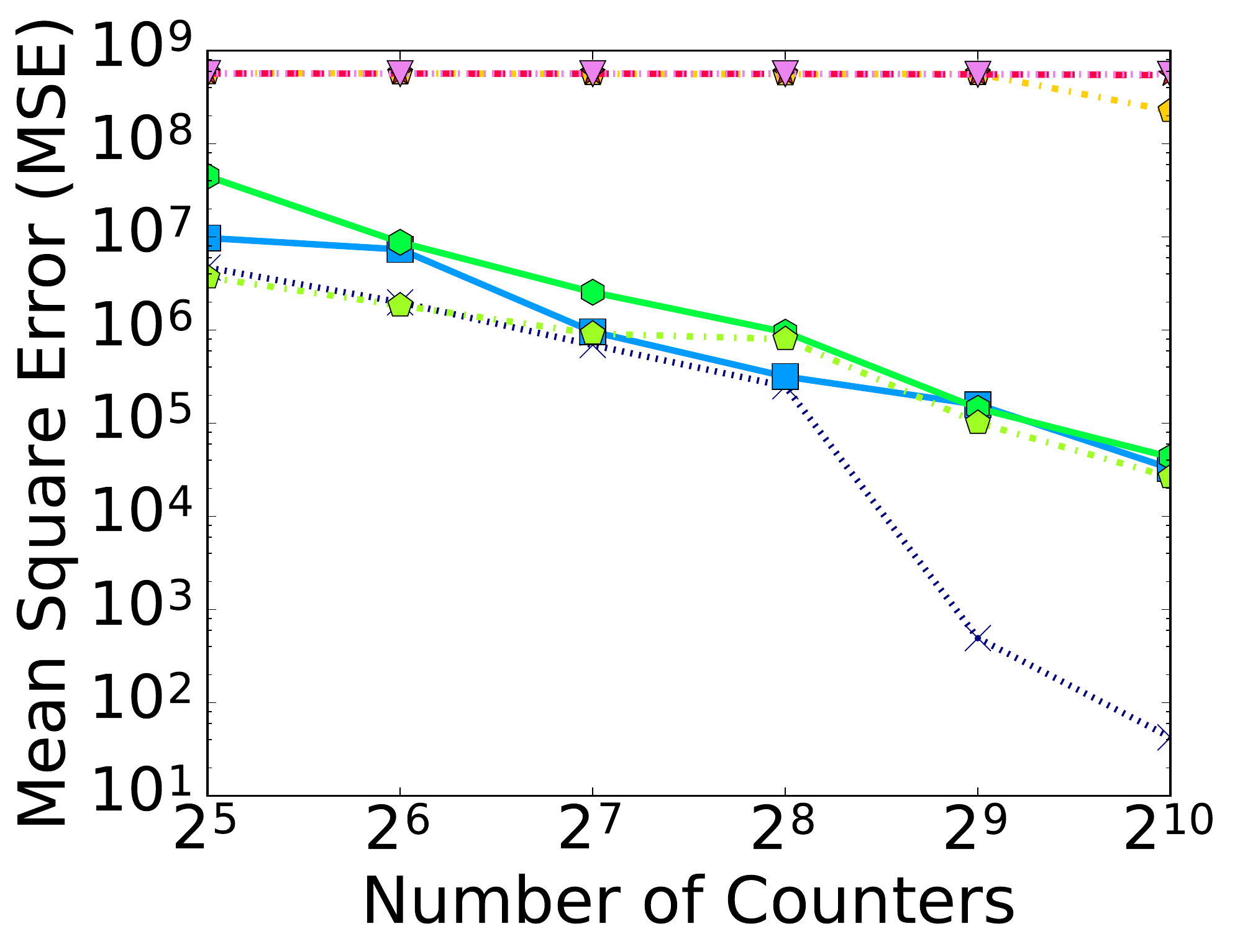}} & 
		\subfloat[UWISC-DC]{\includegraphics[width = 0.66\columnwidth]
			{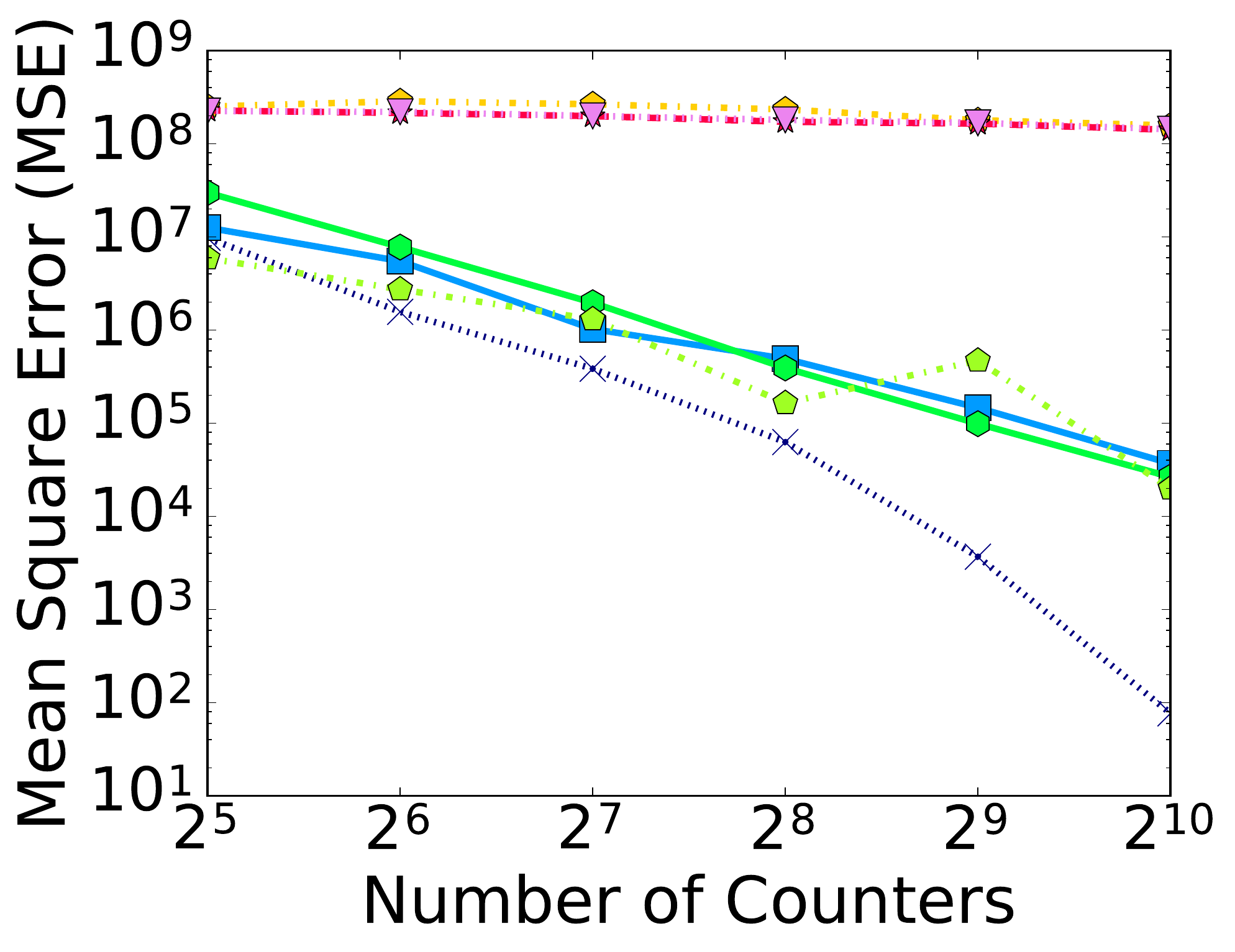}}    
	\end{tabular}    
		\ifdefined\submissionVersion
		\vskip -0.05 cm
		\fi
	\centering{\includegraphics[width = 0.6\linewidth]
		{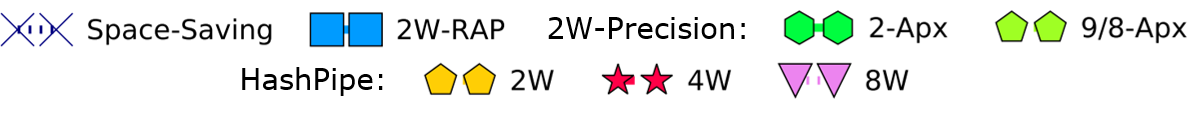}} 
		\ifdefined\submissionVersion
		\vskip -0.4 cm
		\fi
	\begin{tabular}{ccc}
		\subfloat[CAIDA]{\includegraphics[width = 0.66\columnwidth]
			{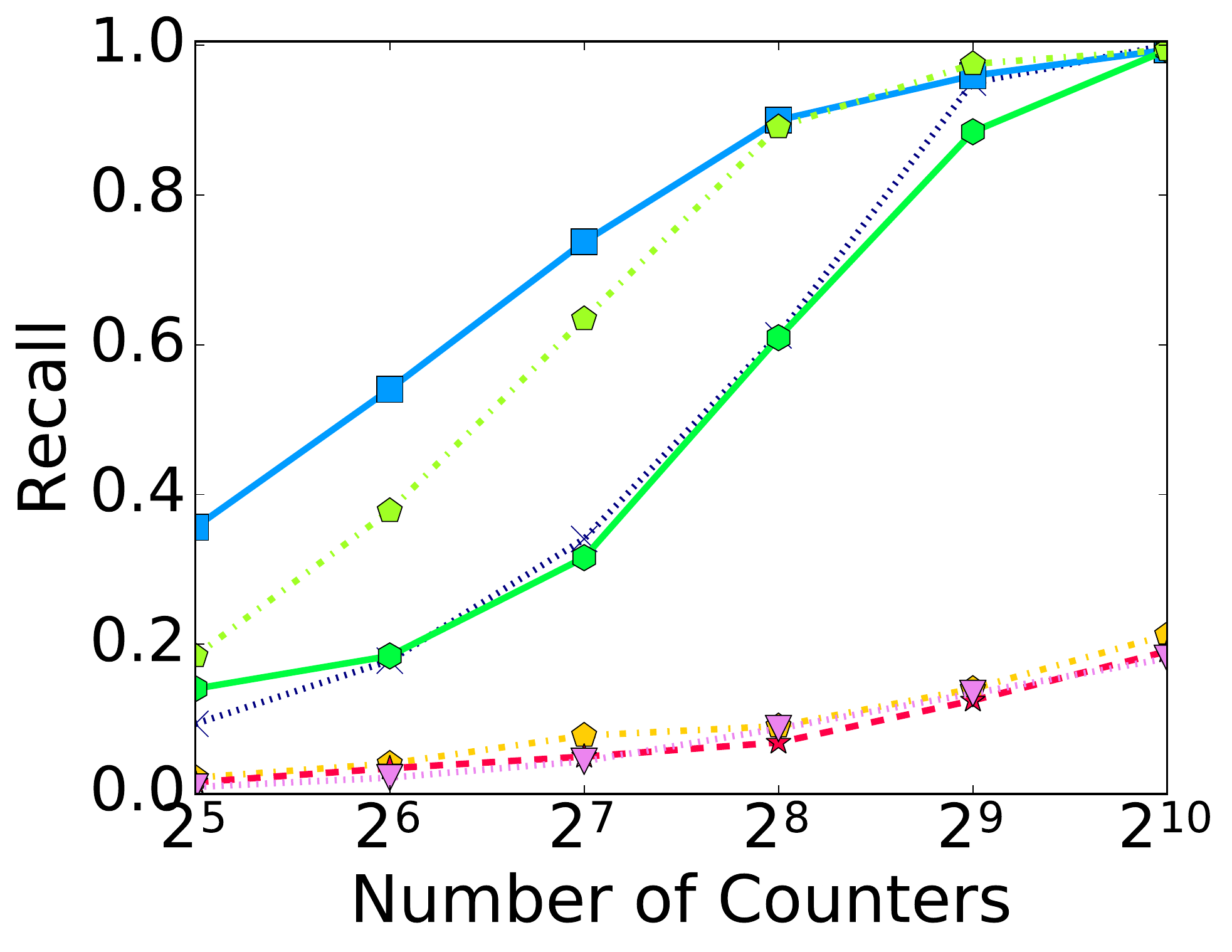}} &    
		\subfloat[UCLA]{\includegraphics[width = 0.66\columnwidth]
			{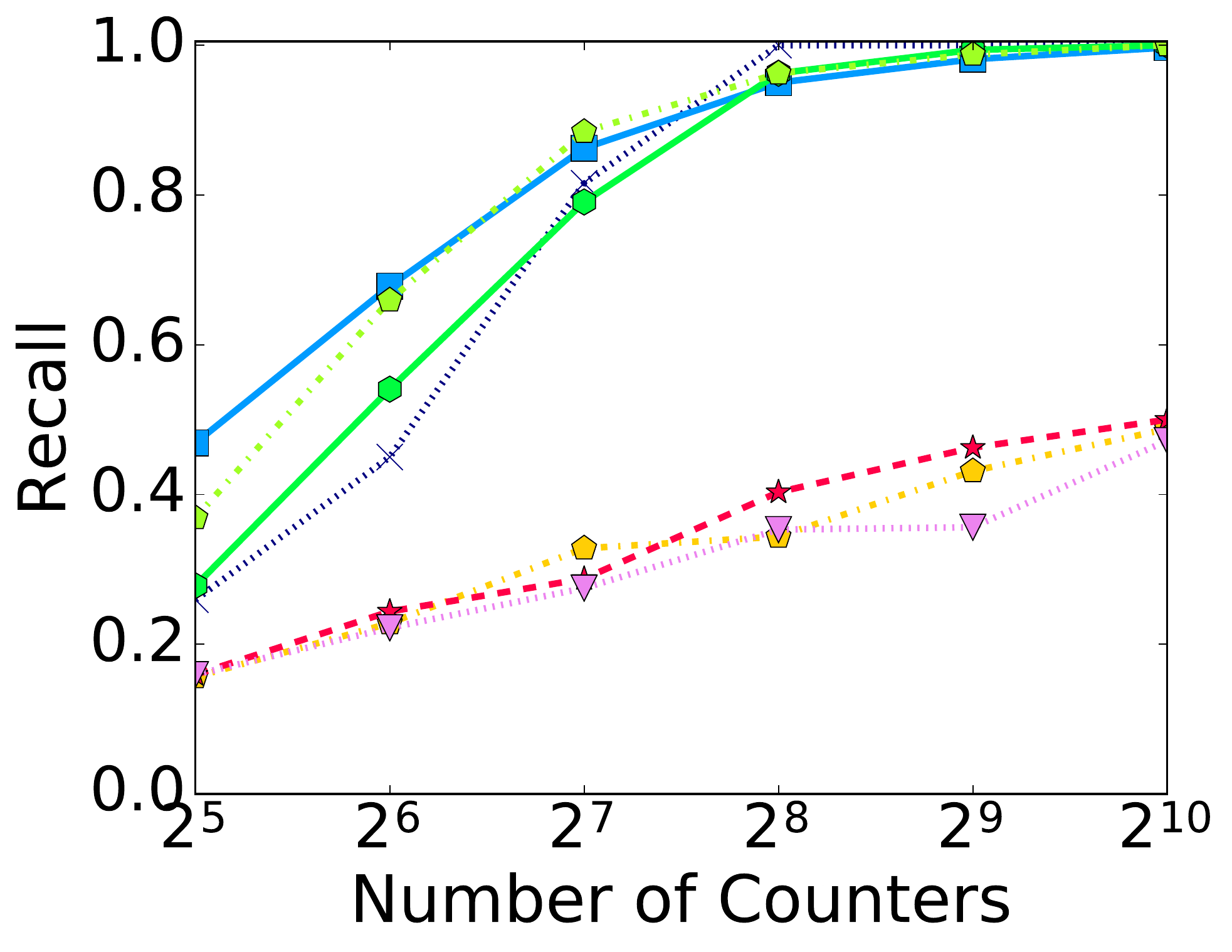}} & 
		\subfloat[UWISC-DC]{\includegraphics[width = 0.66\columnwidth]
			{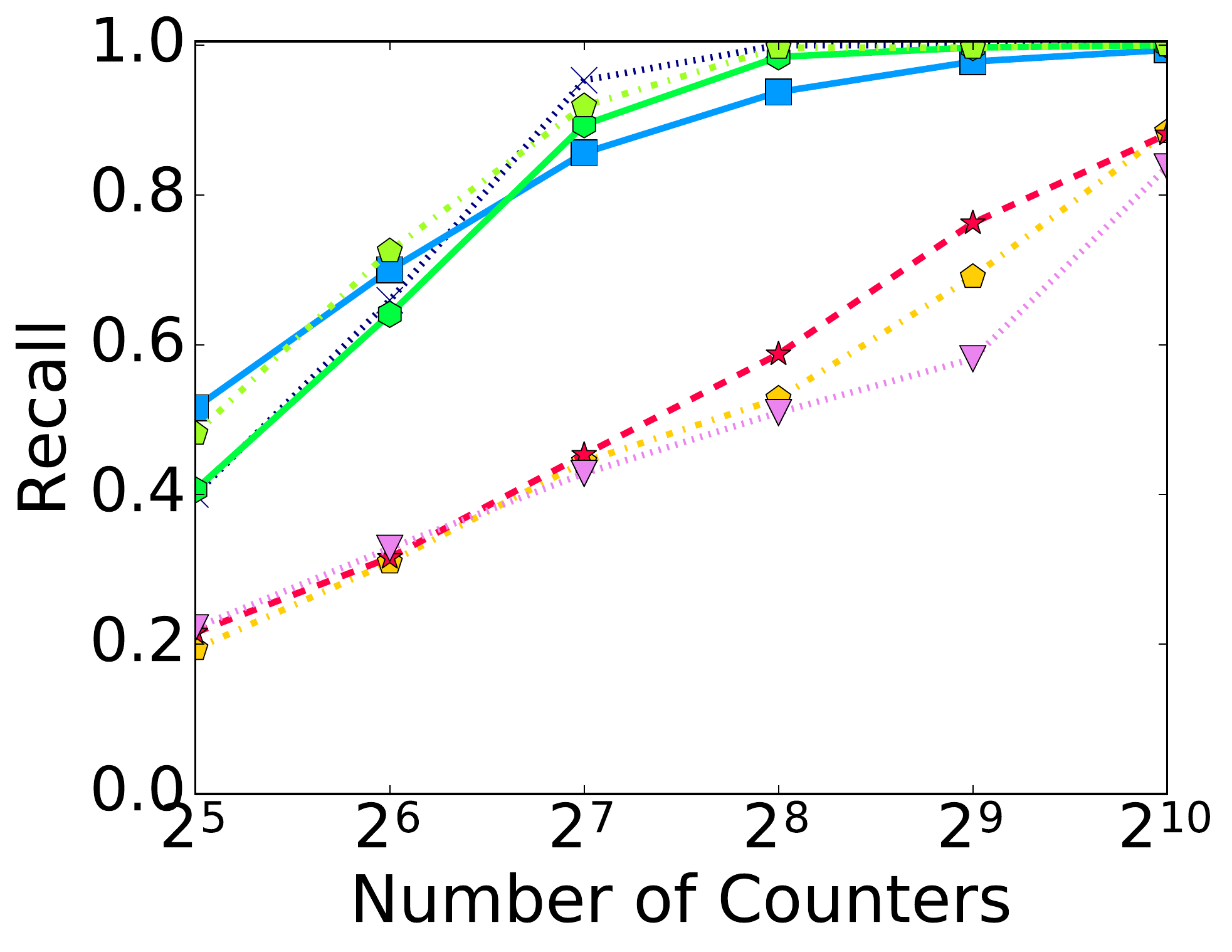}}    
	\end{tabular}    
	\vskip -0.2cm
	\caption{\label{fig:comparativeTopK}Comparative evaluation of the frequency estimation and top-32 problems. }\vskip -0.5cm
\end{figure*}

\subsection{Approximating the desired recirculation probability}
\label{compromize:restricted}

We now evaluate the impact of only using random bits as random source. This limits us to approximate the ideal recirculation probability $\frac{1}{carry\_min+1}$ with a probability of the form $2^{-x}$ or $2^{-y}\times\frac{1}{\lfloor T \rfloor}$. Figure~\ref{fig:RestrictFE} shows results for frequency estimation problem (a) and (b), and for the top-k problem  (c) and (d). 
We evaluated four variants:
``NoAdaptation'' is the algorithm without any hardware-friendly adaptation beyond limited associativity; 
``2-Approximate'' is the variant added with an approximate recirculation probability of $2^{-x}$ form; 
``PRECISION (2-Approximate)'' is the standard PRECISION algorithm with all other hardware-friendly adaptations also added; and
``9/8-Approximate PRECISION'' is the PRECISION algorithm using the $2^{-y}\times\frac{1}{\lfloor T \rfloor}$ form of approximate recirculation probability.

For frequency estimation, the 2-approximation in  recirculation probability increased the error noticeably, possibly due to counters are always bumped to the next power of 2 when replacing a flow entry, causing some overestimation. Meanwhile, using the 9/8-approximation is almost as accurate as having no restriction on the recirculation probability. 

For the top-k problem, we continue with our ongoing evaluation of how many counters are needed to identify the top-32 flows. 
Notice that recirculation probabilities are less impactful in this metric and in the worst case we only need 2$\times$ as many counters as NoAdaptation to achieve the same~Recall.

It is surprising at first to notice that approximating the recirculation probability has a minimal performance impact in the UWISC-DC trace for the top-$k$ metric.
The reason is the highly-concentrated nature of this trace. 
In such workload where heavy hitters dominate, the sizes of tail flows are too small compared with the large counters maintained for heavy hitters, thus the tail flows have little chance to evict heavy hitters regardless of how we approximate probability.

\subsection{Comparison with other algorithms}
\label{eval:comparative}

Next, we evaluate PRECISION with $d=2$ and compare it with Space-Saving~\cite{SpaceSavings}, and HashPipe~\cite{HashPipe} with $d=2, 4, 8$ associativity.  Similarly, we also compare  with a $2$-way set associative RAP~\cite{RAP}.  
Note that RAP was originally designed with a less restrictive programming model, and PRECISION adapts it to the RMT architecture.

Figure~\ref{fig:comparativeTopK} shows results for the frequency estimation and top-$k$ problems. 
Figures~\ref{fig:comparativeTopK}(a)-(c) shows that, for the frequency estimation problem, 2-way RAP and Space-Saving are the most accurate algorithm. They are followed by PRECISION, which is orders of magnitude more accurate than both versions of HashPipe. 

PRECISION requires at most a factor of 4$\times$ increase in memory to match the accuracy of RAP. The performance gap between them is mostly due to approximating the desired recirculation probabilities (recall that PRECISION uses by default the 2-approximation approach). 
Additionally, PRECISION is over 1000x more accurate than different versions of HashPipe, whose performance stagnates as more memory space is provided. Thus, we conclude the frequency estimation evaluation by saying that PRECISION is a dramatic improvement over HashPipe and is not much worse than the state-of-the-art algorithms despite its restricted programming model. 

Figures~\ref{fig:comparativeTopK}(d)-(f) show the Recall performance for the top-$k$ problem. 
In our top-$32$ setup, we see similar trends in all the traces, 
in which the best Recall is achieved by the $2$-way RAP algorithm followed by PRECISION and Space-Saving. The algorithm with the lowest Recall is HashPipe. We see that the $9/8$-approximate probability variant of PRECISION is on par with Space-Saving and not far behind $2$-way RAP.
PRECISION yields similar performance in all traces and requires at most 2$\times$ more space than RAP or Space-Saving. Compared to HashPipe it requires up to 32$\times$ \mbox{less space for the same Recall.}

\section{Conclusions}
\label{sec:conclusions}
This paper outlined the programming capabilities of the recently developed RMT high-performance programmable switch architecture.  
We used the heavy hitter detection problem as an example and exposed the capabilities and restrictions relevant for designing efficient counter-based algorithms within this architecture. 
The need for our study is emphasized as we showed that the previously suggested HashPipe algorithm,
tailored for the P4/PISA pipeline programming model, does not satisfy all the restrictions and may thus be challenging to implement on high-performance programmable switches.

By understanding these restrictions, we introduced PRECISION, a heavy hitter detection algorithm for high-throughput programmable network switches. 
We successfully compiled PRECISION to the newly released Barefoot Tofino switch which is capable of Tbps-scale aggregated throughput. 
PRECISION probabilistically recirculates a small fraction of the packets for a second pipeline traversal. We bounded the amount of recirculation to a small (e.g., 1\%) portion of the packets to achieve a negligible impact on throughput along with competitive accuracy compared to previous algorithms. 

We studied the impact of each RMT architectural restriction on the empirical accuracy. While most restrictions had a minimal effect on the accuracy, the lack of access to an unrestricted random integer generator had a noticeable impact on accuracy. Therefore, we also implemented better approximation for random recirculation probability within the RMT restrictions. We showed that this capability mitigates most of the accuracy loss at the expense of extra hardware resources. 

We performed an extensive evaluation using real and synthetic packet traces.  We showed that PRECISION outperforms recently suggested alternatives for programmable switches~\cite{HashPipe}. Specifically, it is up to 1000$\times$ more accurate when estimating the per-flow frequency and reduces the required space to identify the top 128 flows by a factor of up to 32 compared to HashPipe. We also positioned PRECISION compared to the state-of-the-art software algorithms and showed that it has a similar accuracy compared to the popular Space-Saving algorithm. Moreover, we showed that PRECISION requires at most 4$\times$ more space than RAP when using a naive 2-approximation for recirculation probability, and at most 2$\times$ more space with an improved $9/8$-approximation.

To the best of our knowledge, PRECISION is the first heavy hitter detection algorithm tailored for the RMT architecture. Overall, this work is an important step forward for measurements on high-performance programmable switches as we can now perform heavy-hitter measurements at Tbps-scale aggregated throughput and still benefit from competitive accuracy compared to the state-of-the-art algorithms. 
Further, we hope that our detailed case study of adapting a measurement algorithm to the RMT architecture would be useful for implementing various algorithms on such switches.

\section{Acknowledgments}
This work is supported by NSF grant CCF-1535948. Ran Ben Basat was supported by the Technion HPI research school, Zuckerman Foundation, the Technion Hiroshi Fujiwara cyber security research center and the Israel Cyber Directorate.

We would like to thank our shepherd Patrick~P.~C.~Lee and the anonymous reviewers of ICNP'18 for their helpful feedback.  We also thank Changhoon Kim, Masoud Moshref, Rob Harrison, and Jennifer Rexford for their constructive comments during the writing of this paper.

\ifdefined\submissionVersion
\else
\appendices
\newcommand{\brackets}[1]{\left[#1\right]}
\section{Proof for Bounds on the Amount of Recirculation}


\textbf{Proof for lemma~\ref{lem:iidGeometric}}.
For $n\in\set{1,\ldots,T}$, let $S_n\triangleq \sum_{i=1}^n X_i$ denote the sum of the first $i$ random variables. Next, let $Y\triangleq|\set{1,\ldots,T}\setminus \set{S_n\mid n\in\set{1,\ldots,T-1}}|$ denote the number of integers between $1$ and $T$ that are \emph{not} a sum of prefix of the $X_i$'s. Observe that since the variables $X_i$ are i.i.d., geometric variables, we have that $Y\sim Bin(T-1,p)$; that is, $Y$ is a binomial random variable with mean $p(T-1)$. But observe that the value of $Z$ is simply one plus the number of $n$ values for which $S_n< T$. This establishes that $\mathbb E[Z] = 1+\mathbb E[Y] = p(T-1)+1$.

\textbf{Proof for lemma~\ref{lem:singleCounter}}.
Intuitively, since $\mathbb E[X_i]=i$ and $\sum_{i=1}^n i = O(n^2)$, we can expect that $\Omega(\sqrt T)$ variables are needed to cross the $T$ threshold.
To prove this, first notice that we are looking for an asymptotic bound rather than computing the expectation exactly. 
This allows us to ``ignore'' the first $\sqrt T-1$ random variables. Formally:
\begin{align}
    \mathbb E[A] \le \sqrt T-1 + \mathbb E\brackets{\min\set{{n\in\mathbb N}\mid \sum\limits_{i=\sqrt T}^{\sqrt T+n} X_i\ge T}}.\label{eq1}
\end{align}
Next, let $X_{\sqrt T}',\ldots,X_T'\sim Geo(T^{-1/2})$ denote a set of i.i.d. geometric variables (independent of $X_1,\ldots,X_{\sqrt T}$) with expectation of $\sqrt T$. Notice that for $i\ge \sqrt T$, we have that the parameter for $X_i$ is smaller than that of $X_i'$. Together with~\eqref{eq1}, this allows us to further write:
\begin{align}\label{eq2}
    \mathbb E[A] \le \sqrt T-1 + \mathbb E\brackets{\min\set{{n\in\mathbb N}\mid \sum\limits_{i=\sqrt T}^{\sqrt T+n} X_i'\ge T}}.
\end{align}
Finally, we use Lemma~\ref{lem:iidGeometric} to write $\mathbb E\brackets{\min\set{{n\in\mathbb N}\mid \sum\limits_{i=\sqrt T}^{\sqrt T+n} X_i'\ge T}} = \sqrt T + 1$.
Together with~\eqref{eq2} we conclude that $\mathbb E[A]\le 2\sqrt T$.

\fi

{ \bibliographystyle{acm}
\bibliography{refs}}

\begin{thebibliography}{10}

\bibitem{Tofino}
\url{https://www.barefootnetworks.com/products/brief-tofino/}.

\bibitem{GitHub}
\url{https://github.com/p4lang/p4-applications/tree/master/research_projects/PRECISION}.

\bibitem{CAIDA}
{The CAIDA UCSD Anonymized Internet Traces 2015 - February 19th}.

\bibitem{Arista}
{\sc {Arista Networks}}.
\newblock {Arista 7050X Switch Architecture (‘A day in the life of a
  packet’)}.
\newblock
  \url{https://www.corporatearmor.com/documents/Arista_7050X_Switch_Architecture_Datasheet.pdf}.

\bibitem{RAP}
{\sc Ben{-}Basat, R., Einziger, G., Friedman, R., and Kassner, Y.}
\newblock Randomized admission policy for efficient top-$k$ and frequency
  estimation.
\newblock In {\em IEEE INFOCOM\/} (2017).

\bibitem{IMC10trace}
{\sc Benson, T., Akella, A., and Maltz, D.~A.}
\newblock Network traffic characteristics of data centers in the wild.
\newblock In {\em ACM IMC\/} (2010).

\bibitem{TrafficEngeneering}
{\sc Benson, T., Anand, A., Akella, A., and Zhang, M.}
\newblock Microte: Fine grained traffic engineering for data centers.
\newblock In {\em ACM CoNEXT\/} (2011).

\bibitem{P4}
{\sc Bosshart, P., Daly, D., Gibb, G., Izzard, M., McKeown, N., Rexford, J.,
  Schlesinger, C., Talayco, D., Vahdat, A., Varghese, G., et~al.}
\newblock P4: Programming protocol-independent packet processors.
\newblock {\em ACM SIGCOMM Computer Communication Review\/} (2014).

\bibitem{RMT}
{\sc Bosshart, P., Gibb, G., Kim, H.-S., Varghese, G., McKeown, N., Izzard, M.,
  Mujica, F., and Horowitz, M.}
\newblock Forwarding metamorphosis: {F}ast programmable match-action processing
  in hardware for {SDN}.
\newblock In {\em ACM SIGCOMM Computer Communication Review\/} (2013).

\bibitem{CountSketch}
{\sc Charikar, M., Chen, K., and Farach-Colton, M.}
\newblock Finding frequent items in data streams.
\newblock In {\em EATCS ICALP\/} (2002).

\bibitem{CounterTree}
{\sc Chen, M., and Chen, S.}
\newblock Counter tree: {A} scalable counter architecture for per-flow traffic
  measurement.
\newblock In {\em IEEE ICNP\/} (2015).

\bibitem{dRMT}
{\sc Chole, S., Fingerhut, A., Ma, S., Sivaraman, A., Vargaftik, S., Berger,
  A., Mendelson, G., Alizadeh, M., Chuang, S.-T., Keslassy, I., Orda, A., and
  Edsall, T.}
\newblock {dRMT}: {D}isaggregated programmable switching.
\newblock In {\em ACM SIGCOMM\/} (2017).

\bibitem{SpaceSavingIsTheBest2010}
{\sc Cormode, G., and Hadjieleftheriou, M.}
\newblock Methods for finding frequent items in data streams.
\newblock {\em J. VLDB\/} (2010).

\bibitem{CountMinSketch}
{\sc Cormode, G., and Muthukrishnan, S.}
\newblock An improved data stream summary: The count-min sketch and its
  applications.
\newblock {\em J. Algorithms\/} (2004).

\bibitem{dargahi2017survey}
{\sc Dargahi, T., Caponi, A., Ambrosin, M., Bianchi, G., and Conti, M.}
\newblock A survey on the security of stateful sdn data planes.
\newblock {\em IEEE Communications Surveys \& Tutorials\/} (2017).

\bibitem{LoadBalancing}
{\sc Dittmann, G., and Herkersdorf, A.}
\newblock Network processor load balancing for high-speed links.
\newblock In {\em SPECTS\/} (2002).

\bibitem{BetterNetflow}
{\sc Estan, C., Keys, K., Moore, D., and Varghese, G.}
\newblock Building a better netflow.
\newblock In {\em ACM SIGCOMM\/} (2004).

\bibitem{CUSketch}
{\sc Estan, C., and Varghese, G.}
\newblock New directions in traffic measurement and accounting.
\newblock {\em ACM SIGCOMM\/} (2002).

\bibitem{IntrusionDetection2}
{\sc Garcia-Teodoro, P., Díaz-Verdejo, J.~E., Maciá-Fernández, G., and
  Vázquez, E.}
\newblock Anomaly-based network intrusion detection: Techniques, systems and
  challenges.
\newblock {\em Computers and Security\/} (2009).

\bibitem{FilteredSpaceSaving}
{\sc Homem, N., and Carvalho, J.~P.}
\newblock Finding top-k elements in data streams.
\newblock {\em Inf. Sci.\/} (2010).

\bibitem{ApproximateFairness}
{\sc Kabbani, A., Alizadeh, M., Yasuda, M., Pan, R., and Prabhakar, B.}
\newblock {AF-QCN}: Approximate fairness with quantized congestion notification
  for multi-tenanted data centers.
\newblock In {\em IEEE HOTI\/} (2010).

\bibitem{BatchDecrement}
{\sc Karp, R.~M., Shenker, S., and Papadimitriou, C.~H.}
\newblock A simple algorithm for finding frequent elements in streams and bags.
\newblock {\em ACM Trans. Database Syst.\/} (2003).

\bibitem{UCLA}
{\sc {Laboratory For Advanced Systems Research, UCLA}}.
\newblock \url{http://www.lasr.cs.ucla.edu/ddos/traces/}.

\bibitem{RandomizedCounterSharing}
{\sc Li, T., Chen, S., and Ling, Y.}
\newblock Per-flow traffic measurement through randomized counter sharing.
\newblock {\em IEEE/ACM Trans. on Networking\/} (2012).

\bibitem{FlowRadar}
{\sc Li, Y., Miao, R., Kim, C., and Yu, M.}
\newblock {FlowRadar}: A better netflow for data centers.
\newblock In {\em {USENIX} {NSDI}\/} (2016).

\bibitem{UnivMon}
{\sc Liu, Z., Manousis, A., Vorsanger, G., Sekar, V., and Braverman, V.}
\newblock One sketch to rule them all: Rethinking network flow monitoring with
  {UnivMon}.
\newblock In {\em ACM SIGCOMM\/} (2016).

\bibitem{LC}
{\sc Manku, G.~S., and Motwani, R.}
\newblock Approximate frequency counts over data streams.
\newblock In {\em Int. Conf. on V.L. Data Bases\/} (2002).

\bibitem{SpaceSavings}
{\sc Metwally, A., Agrawal, D., and Abbadi, A.~E.}
\newblock Efficient computation of frequent and top-k elements in data streams.
\newblock In {\em ICDT\/} (2005).

\bibitem{IntrusionDetection}
{\sc Mukherjee, B., Heberlein, L., and Levitt, K.}
\newblock Network intrusion detection.
\newblock {\em IEEE Network\/} (1994).

\bibitem{CounterArray2}
{\sc Ramabhadran, S., and Varghese, G.}
\newblock {Efficient implementation of a statistics counter architecture}.
\newblock {\em ACM SIGMETRICS\/} (2003).

\bibitem{ReversibleSketch}
{\sc Schweller, R., Li, Z., Chen, Y., Gao, Y., Gupta, A., Zhang, Y., Dinda,
  P.~A., Kao, M.-Y., and Memik, G.}
\newblock Reversible sketches: enabling monitoring and analysis over high-speed
  data streams.
\newblock {\em IEEE/ACM Transactions on Networking (ToN) 15}, 5 (2007),
  1059--1072.

\bibitem{CounterArray1}
{\sc Shah, D., Iyer, S., Prabhakar, B., and McKeown, N.}
\newblock Maintaining statistics counters in router line cards.
\newblock {\em IEEE Micro\/} (2002).

\bibitem{Domino}
{\sc Sivaraman, A., Cheung, A., Budiu, M., Kim, C., Alizadeh, M., Balakrishnan,
  H., Varghese, G., McKeown, N., and Licking, S.}
\newblock Packet transactions: High-level programming for line-rate switches.
\newblock In {\em ACM SIGCOMM\/} (2016).

\bibitem{HashPipe}
{\sc Sivaraman, V., Narayana, S., Rottenstreich, O., Muthukrishnan, S., and
  Rexford, J.}
\newblock Heavy-hitter detection entirely in the data plane.
\newblock In {\em ACM SOSR\/} (2017).

\end{thebibliography}
\balance


\end{document}